\documentclass[twocolumn,aps,pra,superscriptaddress,longbibliography]{revtex4-2}

\usepackage{amsmath,amssymb,amsfonts}
\usepackage{amsthm}
\usepackage{color}
\usepackage[dvipsnames]{xcolor}

\usepackage{graphicx}% Include figure files
\usepackage{dcolumn}% Align table columns on decimal point
\usepackage{bm,bbold}% bold math
\usepackage{enumitem}
\usepackage{physics}

\usepackage{tikz}
\usetikzlibrary{quantikz2}

\usepackage{hyperref}

\newtheorem{theorem}{Theorem}
\newtheorem{lemma}{Lemma}

\newcommand{\Gswap}{\mathrm{SWAP}}
\newcommand{\cE}{\mathcal{E}}
\newcommand{\id}{\mathbb{1}}
\newcommand{\peff}{p_\mathrm{eff}}
\newcommand{\pL}{p_\mathrm{L}}

\begin{document}

\title{Fault-tolerant embedding of quantum circuits\\on hardware architectures via swap gates}

\author{Shao-Hen Chiew}
\email{shaohen@entropicalabs.com}
\affiliation{Entropica Labs, 186b Telok Ayer Street, Singapore 068632}

\author{Ezequiel Ignacio Rodríguez Chiacchio}
\affiliation{Entropica Labs, 186b Telok Ayer Street, Singapore 068632}

\author{Vishal Sharma}
\affiliation{Entropica Labs, 186b Telok Ayer Street, Singapore 068632}

\author{Jing Hao Chai}
\affiliation{Entropica Labs, 186b Telok Ayer Street, Singapore 068632}

\author{Hui Khoon Ng}
\email{huikhoon.ng@nus.edu.sg}
\affiliation{Entropica Labs, 186b Telok Ayer Street, Singapore 068632}
\affiliation{Yale-NUS College, Singapore}
\affiliation{Centre for Quantum Technologies, National University of Singapore{, Singapore}}

\date{\today}

\begin{abstract}
In near-term quantum computing devices, connectivity between qubits remain limited by architectural constraints. A computational circuit with given connectivity requirements necessary for multi-qubit gates have to be embedded within physical hardware with fixed connectivity. Long-distance gates have to be done by first routing the relevant qubits together. The simplest routing strategy involves the use of swap gates to swap the information carried by two unconnected qubits to connected ones. Ideal swap gates just permute the qubits; real swap gates, however, have the added possibilities of causing simultaneous errors on the qubits involved and spreading errors across the circuit. A general swap scheme thus changes the error-propagation properties of a circuit, including those necessary for fault-tolerant functioning of a circuit. Here, we present a simple strategy to design the swap scheme needed to embed an abstract circuit onto a physical hardware with constrained connectivity, in a manner that preserves the fault-tolerant properties of the abstract circuit. The embedded circuit will, of course, be noisier, compared to a native implementation of the abstract circuit, but we show in the examples of embedding surface codes on heavy-hexagonal and hexagonal lattices that the deterioration is not severe. This then offers a straightforward solution to implementing circuits with fault-tolerance properties on current hardware.

\end{abstract}

\maketitle

%%%%%%%%%%%%%%%%%%%%%%%%%%
%%%%%%%%%%%%%%%%%%%%%%%%%%
\section{Introduction}
Quantum computing offers the potential to solve problems intractable on classical computers. Full realisation of that potential demands the use of large-scale fault-tolerant quantum hardware \cite{knill1998resilient,preskill1998,aharonov2008,aliferis2005quantum} capable of reliable operations between arbitrary subsets of qubits in the computer. In near-term quantum devices, connectivity between qubits remain limited by architectural constraints, so operations between some subsets of qubits---usually spatially distant ones---cannot be directly implemented \cite{koch2007charge,bravyi2022future,dupont2023quantum,hetenyi2024tailoring}. Instead, one employs a procedure that routes the information carried by those qubits to ones that are connected, perform the necessary operation, and then route the information back. 

Routing procedures, integral to circuit synthesis \cite{wu2021mapping,kremer2024practical} and compilation \cite{Siraichi2018}, include elaborate entanglement-based ones that create a ``portal" between the original and target qubits \cite{choe2024fault}, physical shuttling of the qubits \cite{pino2021demonstration,Bluvstein_2023,stade2024abstract}, and simplistic strategies that use swap gates to effect the information transport along some ``wire" of intervening qubits \cite{finigan2018qubit,li2019tackling,childs2019circuit,nannicini2022optimal,ito2023algorithmic}. Ideal swap gates simply permute the qubits; real swap gates, however, can cause simultaneous errors on the qubits involved, and can spread errors across the circuit. A swap scheme thus generally changes the error-propagation properties of a circuit, and destroys associated fault-tolerance features designed into the original circuit.

Nevertheless, we show here that, by a simple restriction of the allowed swap-moves, we can look for swap-based routing schemes that embed arbitrary circuits onto physical hardware with constrained connectivity in a manner that preserves the fault-tolerant properties. Specifically, the error-patterns from faults in the embedded circuit resemble those of the original circuit; the former thus inherits any fault-tolerant properties---determined by the error-patterns that can be corrected---of the latter. The simplicity of our swap rules make them easy to incorporate into existing circuit-synthesis algorithms. Of course, there will be increased noise in the embedded circuit, due to the added swap gates. We show, however, that in the examples of surface code circuits embedded on heavy-hexagonal and hexagonal lattices, our swap schemes are fairly efficient so that the noise deterioration is mild. Our routing strategy thus offers an immediate route to implementing circuits with fault-tolerant properties on current hardware.

%%%%%%%%%%%%%%%%%%%%%%%%%%
%%%%%%%%%%%%%%%%%%%%%%%%%%
\section{Preliminaries}
Consider a quantum circuit that carries out a specific function---an algorithmic sub-routine, an error-correction procedure, etc. We call this the \emph{abstract circuit}, acting on \emph{abstract qubits}. An \emph{interaction graph} describes the connectivity needed (with a specified temporal order) for the multi-qubit gates in the circuit. We want to implement this circuit on a given hardware, with its \emph{physical qubits}, its set of primitive gates, and its connectivity, i.e., a \emph{device graph} that describes the possible multi-qubit gates between physical qubits.

The task here is to embed the abstract circuit onto the hardware, using swap gates to transport computational information spatially, so that the abstract circuit can be carried out as a sequence of physical operations. Each abstract qubit is assigned to exactly one physical qubit at any point in time. We refer to these physical qubits as the \emph{computational qubits}; the leftover ones are called \emph{routing qubits}. The separation of physical qubits into computational and routing qubits changes as the computation proceeds, as the information carried by the computational qubits is routed via swap gates. The actual circuit carried out on the hardware, including the swap gates and routing qubits, is referred to as the \emph{physical circuit}. \emph{Computational operations} refer to those in the physical circuit that correspond to operations in the abstract circuit. We assume the abstract circuit is already compiled into a circuit that employs only the types of primitive gates available on the hardware, but it may contain multi-qubit gates on abstract qubits assigned to computational qubits not connected on the device graph; such gates are referred to as \emph{long-distance gates}.

In the physical circuit, a long-distance gate from the abstract circuit is realized as a sequence of swaps to bring the computational qubits to connected positions on the device graph so that the multi-qubit gate can be implemented. We can abstract this back into an effective multi-qubit gate on the abstract qubits by tracing away the routing qubits. In the ideal case, this effective gate is exactly equal to the abstract gate; in reality, the swap gates and routing qubits can introduce errors so that the effective gate has a different noise description than a native abstract gate. We refer to the abstraction of the physical circuit back into just action on the abstract qubits, together with the resulting noise, as the \emph{embedded circuit}.

Central to our discussion is the swap gate, a two-qubit linear operation that acts as
\begin{equation}
\Gswap{\left(\ket{\phi}_1\ket{\psi}_2\right)}=\ket{\psi}_1\ket{\phi}_2,
\end{equation}
for arbitrary states $\ket{\phi}$ and $\ket{\psi}$ on the two participating qubits; this extends to arbitrary (non-product) two-qubit states by linearity. An ideal swap gate interchanges the states carried by the two qubits. Any errors in the input state emerge unchanged in the output state, following the respective qubit that brought in the error. The swap can thus be thought of simply as a weaving of the quantum wires, with the state and any errors continuing to be carried by the same wire. A faulty swap gate, however, can introduce new errors. We write the faulty swap gate as the map $\cE\circ\Gswap$, where $\Gswap$ is the ideal gate, and $\cE$ describes the noise channel. $\cE$ can cause errors on one or both qubits involved; it can also spread errors such that incoming errors on one qubit can contaminate the other qubit.

We employ the typical noise model used in discussing fault tolerance of quantum circuits. When a circuit location \cite{gottesman2022opportunities}---a space-time location in the circuit at which a state preparation, a gate (including the identity for a waiting time-step), or a measurement is carried out---functions imperfectly, we say that a \emph{fault} has occurred there. 
A fault on a multi-qubit gate can cause simultaneous errors on one or more of the qubits involved; those errors are allowed to be arbitrary operators (we expect our approach to generalize to biased-noise situations \cite{AP2008,Tuckett2020}, but we leave this for future work). For concreteness, we assume a stochastic model for faults: Faults occur independently on individual circuit locations, each with probability $p$. This can be generalized to a local model, where the independence property is lifted to the constraint that $\ell$ faults occur with probability no larger than $O(p^\ell)$. A \emph{fault-path} refers to a specification of circuit locations with faults. Its weight is the number of faults, so a fault-path of weight $\ell$ occurs with probability $p^\ell$. Each fault-path results in a set of \emph{error-patterns} on the qubits, with each error-pattern specifying the (normalized) error operators applied to the qubits involved in the fault-path. The weight of an error-pattern is the weight of the associated fault-path.

We use this noise model for both physical and abstract circuits. In the latter case, we imagine directly implementing the abstract circuit on the hardware without any SWAPs, disregarding connectivity constraints. This gives us a noisy abstract circuit, against which the performance of our embedded circuit is benchmarked.

%%%%%%%%%%%%%%%%%%%%%%%%%%
%%%%%%%%%%%%%%%%%%%%%%%%%%
\section{Error-pattern-preserving routing}
We want a routing schedule, using swap gates, that implements the interaction graph on the given device graph, with fault-tolerance properties as specified below. The abstract circuit is embedded as a sequence of \emph{computational layers}, during which computational operations occur. Each adjacent pair of computational layers is separated by a sequence of swap gates that route the computational qubits around so that long-distance gates in the next computational layer can be implemented.

We allow for two types of swap gates: (1) Type-1 SWAP-- a swap gate between a computational qubit and a routing qubit; (2) Type-2 SWAP-- a swap gate between two computational qubits that participate in the same operation (or circuit location) in the computational layer either immediately before or after the swap sequence. We refer to that computational operation as one \emph{associated with} the type-2 SWAP. Our routing schedules use only types-1 and 2 SWAPs. We do not use swaps between two routing qubits---these add no useful moves to routing the computational information (though they can change the noise properties of the embedded circuit).

Our main result is to show that an error-pattern in the embedded circuit arising from faulty SWAPs can be mimicked by having an equivalent- or lower-weight fault-path on the computational operations only. The set of error-patterns for the embedded circuit then coincides with that for the noisy abstract circuit; we refer to this as the \emph{error-pattern-preserving} (EPP) property of a routing schedule. 

To show the EPP property, we first consider routing schedules comprising only type-1 SWAPs. Faulty type-1 SWAPs can introduce spatio-temporally correlated errors in the embedded circuit absent in the abstract circuit. This can occur when multiple computational qubits are swapped with the same routing qubit. Nevertheless, the following lemma shows that such correlated errors occur with a probability resembling that of independent errors.

\begin{lemma}\label{lemma}
Consider a weight-$k$ error-pattern in the noisy abstract circuit. If only type-1 SWAPs are allowed, this error-pattern spreads to weight $k+\ell$ in the embedded circuit only if at least $\ell$ type-1 SWAPs are faulty, which occurs with probability $O(p^\ell)$.
\end{lemma}

\begin{proof}
If all SWAPs in the routing schedule are ideal, $\ell=0$, and the weight-$k$ error-pattern is inherited without change by the embedded circuit. A faulty SWAP anywhere, occurring with probability $p$, can introduce errors on the participating computational and routing qubits. The error on the computational qubit adds an error in the error pattern in the embedded circuit. The error on the routing qubit---either generated by the faulty swap gate, or spread by the faulty swap gate from an existing error in the computational qubit---can spread to another computational qubit it is swapped with only if that subsequent SWAP is also faulty. Hence, each additional error on a computational qubit arises only from a faulty SWAP; $\ell$ additional errors require at least $\ell$ faulty SWAPs, which occur with probability $O(p^\ell)$.
\end{proof}
\noindent 
Lemma \ref{lemma} tells us that we can view any error-pattern modified by faulty type-1 SWAPs as if they were generated by additional single-qubit faults, one on each computational qubit involved in a faulty SWAP, occurring with probability $p$. Since we assume arbitrary errors for every fault, we need not worry about how the specific error operator on a computational qubit gets modified in subsequent gates; the added fault on it can simply be absorbed into a fault on the operation it participates in in the nearest computational layer. 

We thus see how an error-pattern in the embedded circuit, in part due to faulty (type-1) SWAPs, can be mimicked by a fault-path on the noisy abstract circuit. For the fault-path in the physical circuit that gave the error-pattern in the embedded circuit, faults on computational operations are assigned to the corresponding operations in the abstract circuit. Each fault on a type-1 SWAP can be viewed as a single-qubit fault on the computational qubit involved. This single-qubit fault is allocated to the operation involving that qubit in the nearest computational layer; the corresponding abstract operation is assigned a fault. Faults allocated to multiple qubits that participate in the same computational operation can be combined into a single fault on that operation, and inherited by the corresponding abstract operation; if that operation is already faulty, any additional faults are absorbed into the existing fault. We thus end up with an abstract-circuit fault-path of equal or lower weight than that of the original physical fault-path. Nevertheless, the error-pattern in the embedded circuit can be found as one of the error patterns arising from this mocked-up fault-path in the abstract circuit. 

It is straightforward to extend this to routing schedules with both types-1 and 2 SWAPs, by first viewing the type-2 SWAPs as if they are computational gates. In this case, we first absorb the faults from the type-1 SWAPs into either the nearest type-2 SWAP or the computational layer in the manner described above. We are then left only with faulty type-2 SWAPs immediately adjacent to computational layers, after removing the no-fault type-1 SWAPs (which act now as ideal permutations). The fault on each type-2 SWAP can then be absorbed into its associated computational operation, and the corresponding abstract operation inherits that fault. 

The above construction gives us the following Theorem about the EPP nature of our routing schedules:
\begin{theorem}[EPP routing schedules]\label{theorem}
Consider an abstract circuit that employs a routing schedule with only types-1 and 2 SWAPs. For every error-pattern that arises in the embedded circuit from a fault-path in the physical circuit, there exists an equal- or lower-weight fault-path in the noisy abstract circuit that can give the same error-pattern. 
\end{theorem}

Theorem \ref{theorem} allows us to argue that our routing schedules preserve the fault-tolerance properties of an abstract circuit. A quantum circuit, built upon an error-correcting code, is fault tolerant if it does not convert correctable errors into uncorrectable ones, even in the presence of faults in the circuit. Concrete conditions can be enforced to assure proper behavior of a circuit when no more than $t$ faults occur in the circuit, for a scheme built upon a code that corrects errors on $t$ or fewer qubits (see, for example, Refs.~\cite{aliferis2005quantum,gottesman2013fault,delfosse2020short,tansuwannont2023adaptive}). All such conditions involve checks on the error-patterns in the circuit arising from fault-paths of weights $t$ or smaller. Our EPP routing schedule preserves the set of such error-patterns according to Theorem \ref{theorem}, and hence the embedded circuit automatically inherits the fault-tolerance properties of the abstract circuit. Put differently, if the abstract circuit is such that $t$ or fewer faults anywhere lead only to correctable error-patterns, the physical circuit with $t$ or fewer faults also results in an embedded circuit with only correctable error-patterns.

Of course, the embedded circuit will be noisier than the noisy abstract circuit---the physical circuit, with the added swap gates, has many more ways in which faults can occur. The number of ways in which more than $t$ faults occur in the physical circuit will hence generally be larger than that for the abstract circuit. These $>\!t$-fault situations give the failure modes of the abstract or embedded circuit, and determine the resulting logical error rate. The noisier embedded circuit will thus have a larger logical error rate, which translates into a more stringent fault-tolerance noise threshold below which errors can be removed by noisy error correction. We illustrate these points in our surface-code examples below.

We have thus reduced the problem of fault-tolerantly embedding an abstract circuit on physical hardware to a search of an EPP routing schedule. Embedding an abstract circuit onto a given hardware topology falls under the general problem of compilation of quantum circuits using swap gates. This is well studied in the literature \cite{finigan2018qubit,li2019tackling,childs2019circuit,nannicini2022optimal,ito2023algorithmic}, but current routing algorithms are insensitive to fault-tolerance properties. By adding the EPP constraints of limiting to types-1 and 2 SWAPS, we are able to easily adapt existing routing search algorithms; see App.~\ref{app:Routing}.

%%%%%%%%%%%%%%%%%%%%%%%%%%
%%%%%%%%%%%%%%%%%%%%%%%%%%
\section{Surface code examples}

As an example, we apply our results to embed the error correction circuit for the rotated planar surface code \cite{fowler2012surface,tomita2014low}---a popular approach to error correction---on hardware with heavy-hexagonal (as used in IBM Quantum platforms) and hexagonal lattices as their device graphs. The rotated planar surface code has data qubits on the vertices of a square lattice, and $X$- and $Z$-type ancillary qubits at the centre of each plaquette in a checkerboard arrangement. The $X$ and $Z$ ancillas are used in the syndrome extraction (SE) circuits shown in Fig.~\ref{fig:SC}(b). Syndrome extraction happens across the entire surface code lattice simultaneously, with the $X$ and $Z$ SE circuits running in parallel and the CNOTs connecting each ancilla with data qubits in a specified order \cite{tomita2014low}. Each SE round thus comprises an ancilla-preparation layer, 4 CNOT layers, and finally a measurement layer. Our abstract circuit here implements a full error-correction cycle for a distance-$d$ surface code, comprising $d$ rounds of syndrome extraction using the SE circuits, followed by syndrome decoding to infer and (attempt to) correct the errors. The CNOTs in the SE circuits give the interaction graph a (rotated) square-lattice structure, a mismatch with the heavy-hexagonal and hexagonal device graphs. 

\begin{figure}
\includegraphics[trim=0mm 0mm 0mm 0mm, clip, width=\columnwidth]{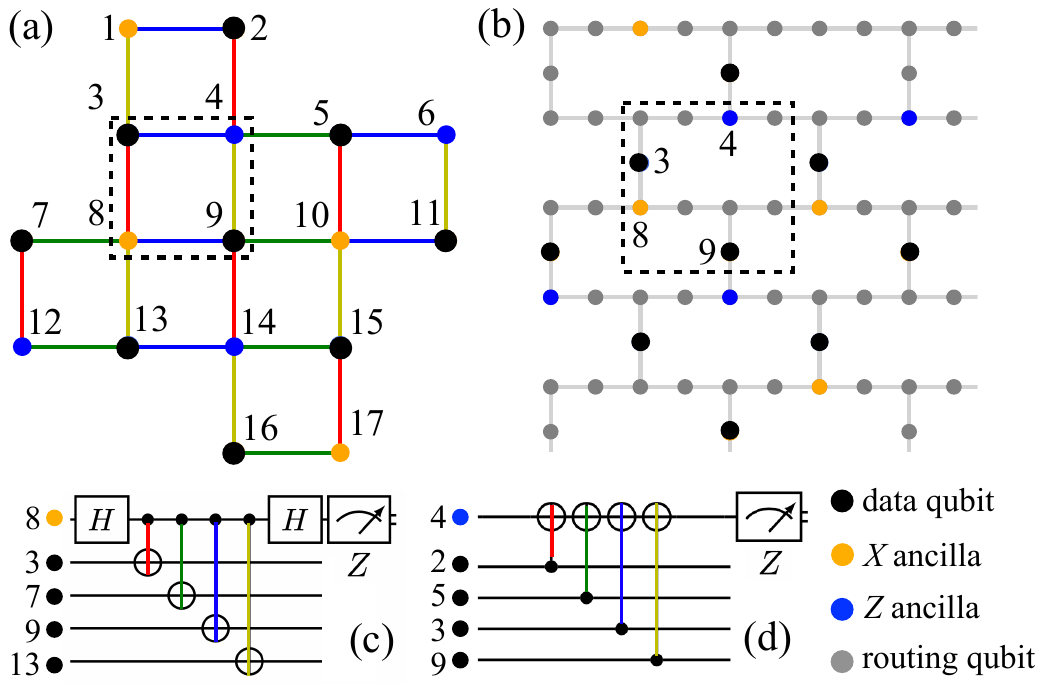}
\caption{\label{fig:SC} The surface code examples. (a) The interaction graph of the rotated planar surface code, shown for the distance-3 case, with each edge colored according to the temporal layer in which a CNOT occurs between the data and ancilla qubits. The unit cell is enclosed in the dotted box. (b) The initial placement of the abstract surface-code lattice onto the heavy-hexagonal device graph. The corresponding unit cell is enclosed in the dotted box. The analogous placement for the hexagonal case can be found in App.~\ref{app:HexEx}. (c) and (d) show the $X$ and $Z$ syndrome extraction circuits, respectively, with the qubit numbering and coloring of the CNOT gates matching those in (a).
}
\end{figure}

The periodic natures of the interaction and device graphs allow us to reduce our embedding problem to consider only unit cells of the lattices. This same unit cell works for all distances of the surface code, as the code grows only in area while retaining the local structures; the same routing schedule thus works for all distances. We choose an initial placement of the abstract unit cell onto the device-graph unit cell, i.e., assign each abstract qubit to a physical qubit, dividing the physical qubits into an initial arrangement of computational and routing qubits; see Fig.~\ref{fig:SC}(c). With these initial layouts, we search for EPP routing schedules as described in App.~\ref{app:SCRouting2}, looking for ones with minimal depth. The surface-code structure allows for a simple parameterization of routing schedules, making for an easy search. 

For the heavy-hexagonal situation, the minimal-depth solution found by our algorithm
comprises a total of five SWAP layers per SE round, with each physical qubit (excluding the altogether unused ones) participating in 3.5 SWAPs on average (see App.~\ref{app:HeavyEx}). At the end of one SE round, the computational qubits are displaced one unit cell diagonally from the initial layout positions; in the next round, the directions and ordering of the routing moves are reversed---allowed by the symmetry of the lattices---so that the computational qubits return to the initial positions after yet another cycle, to prevent walk-off of the embedded circuit. The minimal-depth EPP schedule for the hexagonal situation is much simpler: One needs only a single type-2 SWAP layer between the final two CNOT layers for each SE round. 

To illustrate the EPP nature of our routing schedules, we perform numerical simulations (using Stim \cite{gidney2021stim}) of the physical circuits for different surface-code distances $d$ under circuit-level depolarizing noise characterized by error probability $p$. We demonstrate the appearance of a fault-tolerance threshold in a manner closely resembling that of the abstract circuit, apart from a noise rescaling explained below. To imitate the standard situation where one type of two-qubit gate is available, each SWAP is composed from three consecutive CNOT gates, needed also in the SE circuits. Pauli errors are inserted at all layers of the circuit: initialization to $\ket 0$ and $Z$-measurement suffer bit-flip $X$ with probability $p$; non-identity single-qubit gates have Pauli $X$, $Y$, or $Z$ errors, each with probability $p/3$; CNOT gates experience non-identity two-qubit errors $\sigma_i\otimes \sigma_j$, for $\sigma_i,\sigma_j=\id, X,Y,Z$, each with probability $p/15$. Identity gates are treated as ideal. This circuit-level noise is applied to the physical circuits with our EPP routing for different $d$, and we benchmark the logical error probability $\pL$ at the end of the error-correction cycle against that for an abstract circuit with the same circuit-level noise. For syndrome decoding, we employ PyMatching \cite{higgott2022pymatching,higgott2023sparse} based on the minimum-weight-perfect-matching (MWPM) algorithm, which remains applicable here (see App.~\ref{sec:mwpm_correction}).

The results for the heavy-hexagonal case are given in Fig.~\ref{fig:SCresults}, with $\pL$ plotted against the effective error probability $\peff$, equal to $p$ for the abstract circuit, but rescaled to the best-fit value of $\peff=3.63p$ for the embedded circuit. This $\peff$ should be thought of as the effective error probability for the embedded circuit, absorbing the noise introduced by the extra swap gates into an effective noise on the abstract qubits. The horizontal rescaling is done to bring the $\pL$ curves for the the abstract and embedded circuits on top of each other for easier comparison. We see that, for each $d$, the gradients of the curves for the abstract and embedded circuits match very well, demonstrating the EPP nature of our routing schedule: The gradient of the $\pL$ versus $\peff$ (or $p$) line tells us the number of faults the circuit is tolerant to, and a preservation of the fault-tolerant nature of a circuit demands precisely an unchanged gradient. Consequently, the lines for different $d$ also intersect at the same point on the $\pL$ versus $\peff$ graph, for both abstract and embedded circuits. 

We thus conclude that our EPP routing schedule preserves the fault-tolerant nature of the surface code error correction circuit, with a worsening of the threshold by a factor of 3.63; this factor will be smaller if each SWAP is a single high-fidelity operation, rather than the three CNOTs assumed here.
An analogous simulation of the hexagonal situation gives similar results, with $\peff=1.25p$ in that case; see App.~\ref{app:HexEx}. A straightforward derivation (see App.~\ref{sec:approx_noise_model_surf}) yields an approximate expression for $\peff$, giving $\peff\simeq 3.1p$ and $1.25p$, respectively, for the heavy-hexagonal and hexagonal cases. Deviations from the best-fit values arise from taking a spatio-temporal average of the routing schedule parameters, and from ignoring the asymmetry in the Pauli errors introduced by type-1 SWAPs.

Compared with recent proposals for implementing the surface code on hexagonal \cite{mcewen2023relaxing} and  heavy-hexagonal lattices \cite{benito2024comparative,hetenyi2024creating}, we see a threshold deterioration of the same order of magnitude with the same spatial footprint. Notably though, those schemes relied on a modification of the standard surface code SE circuits to have a \emph{hexagonal} interaction graph before the embedding; our approach requires no such modification and is applicable to arbitrary circuits and device graphs for which alternative interaction graphs may not be possible.

\begin{figure}
\includegraphics[trim=0mm 0mm 15mm 12mm, clip, width=\columnwidth]{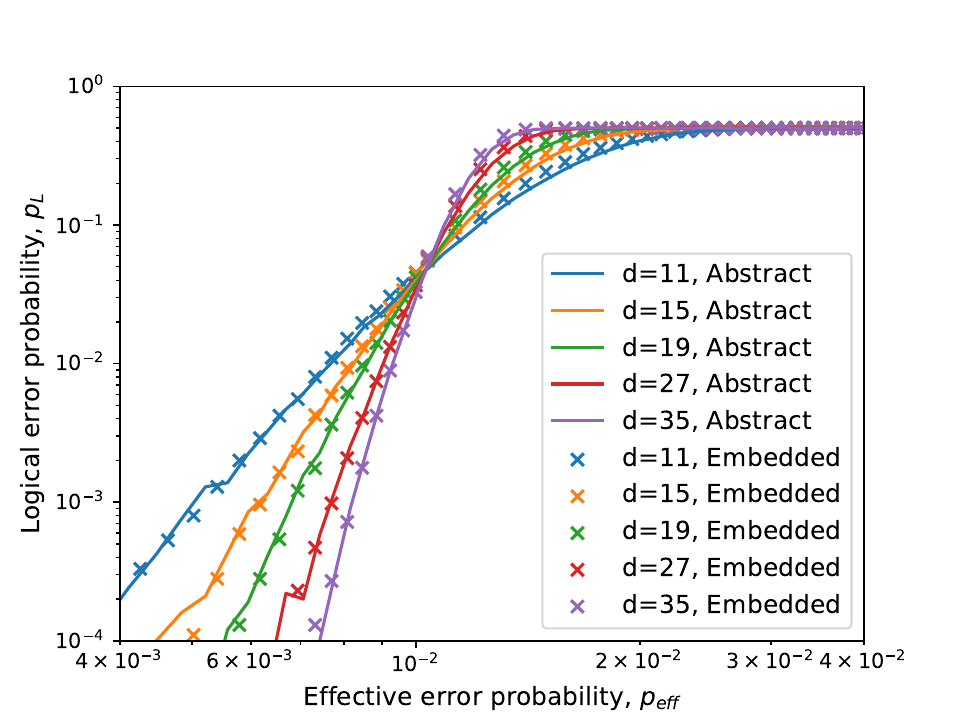}
\caption{\label{fig:SCresults} The logical error probability $\pL$ versus the effective error probability $\peff$ for the example of the surface code embedded onto the heavy-hexagonal lattice under depolarizing noise with error probability $p$. For the benchmark noisy abstract circuit, $\peff=p$, the Pauli-error probability; for the embedded circuit, $\peff=3.63p$, obtained by fitting both curves for $d=35$ below the threshold value.
}
\end{figure}

%%%%%%%%%%%%%%%%%%%%%%%%%%
%%%%%%%%%%%%%%%%%%%%%%%%%%
\section{Discussion and Outlook}
We have shown that a restriction to types-1 and 2 SWAPs results in EPP routing schedules, thus preserving fault-tolerance properties of an abstract circuit. Noisy type-1 SWAPs can in fact be thought of as interactions between the computational qubits with a possibly spatially and temporally correlated environment formed by the routing qubits; the routing qubits move around as SWAPs occur, but they never participate in the computational operations. From this perspective, that type-1 SWAPs preserve fault-tolerance properties becomes obvious: Standard fault-tolerance requirements allow for arbitrary, though limited (by the code correction capacity), interactions with an environment \cite{terhal2005fault,aliferis2005quantum,aharonov2006fault,ng2009fault}. Type-2 SWAPs imitate the errors introduced by an existing interaction in the abstract circuit. These two observations give the intuition behind the fault-tolerance-preserving nature of our routing approach.

The restriction to types-1 and 2 SWAPs is easy to accommodate within existing routing algorithms. In our surface-code examples, the problem symmetries further enabled a concise search for the optimal routing schedule. Relaxing the restriction to types-1 and 2 SWAPs might be possible for surface codes, by including the expanded set of error-patterns within the decoding algorithm, e.g., adding hypergraph edges in the syndrome lattice \cite{higgott2023improved,chen2022calibrated}, but we leave this as a possibility for future investigations.

Our surface-code examples show a mild noise increase in the embedded circuits from the added SWAPs. The swap scheme can possibly be further optimized for lower depth---at least for the heavy-hexagonal situation---by choosing a different initial embedding. Our results do not guarantee, for the surface-code examples or more general situations, the existence of EPP routing schedules for given interaction and device graphs. One might expect a sparser embedding with more routing qubits to yield EPP schedules more readily, but at the cost of a larger spatial footprint.

We note that our current routing search algorithm, for simplicity, constrains SWAPs and computational gates into separate circuit layers. For our heavy-hexagonal example, we found an EPP routing schedule with 6 SWAP layers per SE round; upon inspection, however, the circuit can be compressed---by allowing swap-gate CNOTs and computational gates to occupy the same layers---into a circuit depth equal to that of the minimal-depth solution discussed above. Interestingly, this 6-SWAP-layer solution introduces only temporally, but not spatially, correlated errors, a property that may be helpful in limiting the spread of errors. A more flexible circuit compilation program that allows for simultaneous circuit compression and EPP routing search can potentially generate more solutions, including ones with desirable features.

All in all, our work relaxes hardware connectivity requirements for near-term fault-tolerant tasks. Its general applicability provides a concrete route to implementing arbitrary circuits on hardware in a fault-tolerant manner, without having to redesign the abstract circuit. Our work also motivates the development of native, high-quality swap gates, such that the information transport itself gives minimal increase in the overall noise. Perhaps, we can eventually achieve with swap gates, the high-fidelity information transport possible today in neutral-atom experiments via optical tweezers \cite{Bluvstein_2023}.

%%%%%%%%%%%%%%%%%%%%
\begin{acknowledgments}
HK Ng acknowledges partial support from the  National Research Foundation, Singapore and A*STAR under its CQT Bridging Grant.
\end{acknowledgments}

\bibliography{apssamp}

\begin{thebibliography}{42}%
\makeatletter
\providecommand \@ifxundefined [1]{%
 \@ifx{#1\undefined}
}%
\providecommand \@ifnum [1]{%
 \ifnum #1\expandafter \@firstoftwo
 \else \expandafter \@secondoftwo
 \fi
}%
\providecommand \@ifx [1]{%
 \ifx #1\expandafter \@firstoftwo
 \else \expandafter \@secondoftwo
 \fi
}%
\providecommand \natexlab [1]{#1}%
\providecommand \enquote  [1]{``#1''}%
\providecommand \bibnamefont  [1]{#1}%
\providecommand \bibfnamefont [1]{#1}%
\providecommand \citenamefont [1]{#1}%
\providecommand \href@noop [0]{\@secondoftwo}%
\providecommand \href [0]{\begingroup \@sanitize@url \@href}%
\providecommand \@href[1]{\@@startlink{#1}\@@href}%
\providecommand \@@href[1]{\endgroup#1\@@endlink}%
\providecommand \@sanitize@url [0]{\catcode `\\12\catcode `\$12\catcode `\&12\catcode `\#12\catcode `\^12\catcode `\_12\catcode `\%12\relax}%
\providecommand \@@startlink[1]{}%
\providecommand \@@endlink[0]{}%
\providecommand \url  [0]{\begingroup\@sanitize@url \@url }%
\providecommand \@url [1]{\endgroup\@href {#1}{\urlprefix }}%
\providecommand \urlprefix  [0]{URL }%
\providecommand \Eprint [0]{\href }%
\providecommand \doibase [0]{https://doi.org/}%
\providecommand \selectlanguage [0]{\@gobble}%
\providecommand \bibinfo  [0]{\@secondoftwo}%
\providecommand \bibfield  [0]{\@secondoftwo}%
\providecommand \translation [1]{[#1]}%
\providecommand \BibitemOpen [0]{}%
\providecommand \bibitemStop [0]{}%
\providecommand \bibitemNoStop [0]{.\EOS\space}%
\providecommand \EOS [0]{\spacefactor3000\relax}%
\providecommand \BibitemShut  [1]{\csname bibitem#1\endcsname}%
\let\auto@bib@innerbib\@empty
%</preamble>
\bibitem [{\citenamefont {Knill}\ \emph {et~al.}(1998)\citenamefont {Knill}, \citenamefont {Laflamme},\ and\ \citenamefont {Zurek}}]{knill1998resilient}%
  \BibitemOpen
  \bibfield  {author} {\bibinfo {author} {\bibfnamefont {E.}~\bibnamefont {Knill}}, \bibinfo {author} {\bibfnamefont {R.}~\bibnamefont {Laflamme}},\ and\ \bibinfo {author} {\bibfnamefont {W.~H.}\ \bibnamefont {Zurek}},\ }\bibfield  {title} {\bibinfo {title} {Resilient quantum computation},\ }\href@noop {} {\bibfield  {journal} {\bibinfo  {journal} {Science}\ }\textbf {\bibinfo {volume} {279}},\ \bibinfo {pages} {342} (\bibinfo {year} {1998})}\BibitemShut {NoStop}%
\bibitem [{\citenamefont {Preskill}(1998)}]{preskill1998}%
  \BibitemOpen
  \bibfield  {author} {\bibinfo {author} {\bibfnamefont {J.}~\bibnamefont {Preskill}},\ }\bibfield  {title} {\bibinfo {title} {Reliable quantum computers},\ }\href {https://doi.org/10.1098/rspa.1998.0167} {\bibfield  {journal} {\bibinfo  {journal} {Proceedings of the Royal Society of London. Series A: Mathematical, Physical and Engineering Sciences}\ }\textbf {\bibinfo {volume} {454}},\ \bibinfo {pages} {385} (\bibinfo {year} {1998})}\BibitemShut {NoStop}%
\bibitem [{\citenamefont {Aharonov}\ and\ \citenamefont {Ben-Or}(2008)}]{aharonov2008}%
  \BibitemOpen
  \bibfield  {author} {\bibinfo {author} {\bibfnamefont {D.}~\bibnamefont {Aharonov}}\ and\ \bibinfo {author} {\bibfnamefont {M.}~\bibnamefont {Ben-Or}},\ }\bibfield  {title} {\bibinfo {title} {Fault-tolerant quantum computation with constant error rate},\ }\href {https://doi.org/10.1137/S0097539799359385} {\bibfield  {journal} {\bibinfo  {journal} {SIAM Journal on Computing}\ }\textbf {\bibinfo {volume} {38}},\ \bibinfo {pages} {1207} (\bibinfo {year} {2008})},\ \Eprint {https://arxiv.org/abs/https://doi.org/10.1137/S0097539799359385} {https://doi.org/10.1137/S0097539799359385} \BibitemShut {NoStop}%
\bibitem [{\citenamefont {Aliferis}\ \emph {et~al.}(2005)\citenamefont {Aliferis}, \citenamefont {Gottesman},\ and\ \citenamefont {Preskill}}]{aliferis2005quantum}%
  \BibitemOpen
  \bibfield  {author} {\bibinfo {author} {\bibfnamefont {P.}~\bibnamefont {Aliferis}}, \bibinfo {author} {\bibfnamefont {D.}~\bibnamefont {Gottesman}},\ and\ \bibinfo {author} {\bibfnamefont {J.}~\bibnamefont {Preskill}},\ }\bibfield  {title} {\bibinfo {title} {Quantum accuracy threshold for concatenated distance-3 codes},\ }\href@noop {} {\bibfield  {journal} {\bibinfo  {journal} {arXiv preprint quant-ph/0504218}\ } (\bibinfo {year} {2005})}\BibitemShut {NoStop}%
\bibitem [{\citenamefont {Koch}\ \emph {et~al.}(2007)\citenamefont {Koch}, \citenamefont {Terri}, \citenamefont {Gambetta}, \citenamefont {Houck}, \citenamefont {Schuster}, \citenamefont {Majer}, \citenamefont {Blais}, \citenamefont {Devoret}, \citenamefont {Girvin},\ and\ \citenamefont {Schoelkopf}}]{koch2007charge}%
  \BibitemOpen
  \bibfield  {author} {\bibinfo {author} {\bibfnamefont {J.}~\bibnamefont {Koch}}, \bibinfo {author} {\bibfnamefont {M.~Y.}\ \bibnamefont {Terri}}, \bibinfo {author} {\bibfnamefont {J.}~\bibnamefont {Gambetta}}, \bibinfo {author} {\bibfnamefont {A.~A.}\ \bibnamefont {Houck}}, \bibinfo {author} {\bibfnamefont {D.~I.}\ \bibnamefont {Schuster}}, \bibinfo {author} {\bibfnamefont {J.}~\bibnamefont {Majer}}, \bibinfo {author} {\bibfnamefont {A.}~\bibnamefont {Blais}}, \bibinfo {author} {\bibfnamefont {M.~H.}\ \bibnamefont {Devoret}}, \bibinfo {author} {\bibfnamefont {S.~M.}\ \bibnamefont {Girvin}},\ and\ \bibinfo {author} {\bibfnamefont {R.~J.}\ \bibnamefont {Schoelkopf}},\ }\bibfield  {title} {\bibinfo {title} {Charge-insensitive qubit design derived from the cooper pair box},\ }\href@noop {} {\bibfield  {journal} {\bibinfo  {journal} {Physical Review A}\ }\textbf {\bibinfo {volume} {76}},\ \bibinfo {pages} {042319} (\bibinfo {year} {2007})}\BibitemShut {NoStop}%
\bibitem [{\citenamefont {Bravyi}\ \emph {et~al.}(2022)\citenamefont {Bravyi}, \citenamefont {Dial}, \citenamefont {Gambetta}, \citenamefont {Gil},\ and\ \citenamefont {Nazario}}]{bravyi2022future}%
  \BibitemOpen
  \bibfield  {author} {\bibinfo {author} {\bibfnamefont {S.}~\bibnamefont {Bravyi}}, \bibinfo {author} {\bibfnamefont {O.}~\bibnamefont {Dial}}, \bibinfo {author} {\bibfnamefont {J.~M.}\ \bibnamefont {Gambetta}}, \bibinfo {author} {\bibfnamefont {D.}~\bibnamefont {Gil}},\ and\ \bibinfo {author} {\bibfnamefont {Z.}~\bibnamefont {Nazario}},\ }\bibfield  {title} {\bibinfo {title} {The future of quantum computing with superconducting qubits},\ }\href@noop {} {\bibfield  {journal} {\bibinfo  {journal} {Journal of Applied Physics}\ }\textbf {\bibinfo {volume} {132}} (\bibinfo {year} {2022})}\BibitemShut {NoStop}%
\bibitem [{\citenamefont {Dupont}\ \emph {et~al.}(2023)\citenamefont {Dupont}, \citenamefont {Evert}, \citenamefont {Hodson}, \citenamefont {Sundar}, \citenamefont {Jeffrey}, \citenamefont {Yamaguchi}, \citenamefont {Feng}, \citenamefont {Maciejewski}, \citenamefont {Hadfield}, \citenamefont {Alam}, \citenamefont {Wang}, \citenamefont {Grabbe}, \citenamefont {Lott}, \citenamefont {Rieffel}, \citenamefont {Venturelli},\ and\ \citenamefont {Reagor}}]{dupont2023quantum}%
  \BibitemOpen
  \bibfield  {author} {\bibinfo {author} {\bibfnamefont {M.}~\bibnamefont {Dupont}}, \bibinfo {author} {\bibfnamefont {B.}~\bibnamefont {Evert}}, \bibinfo {author} {\bibfnamefont {M.~J.}\ \bibnamefont {Hodson}}, \bibinfo {author} {\bibfnamefont {B.}~\bibnamefont {Sundar}}, \bibinfo {author} {\bibfnamefont {S.}~\bibnamefont {Jeffrey}}, \bibinfo {author} {\bibfnamefont {Y.}~\bibnamefont {Yamaguchi}}, \bibinfo {author} {\bibfnamefont {D.}~\bibnamefont {Feng}}, \bibinfo {author} {\bibfnamefont {F.~B.}\ \bibnamefont {Maciejewski}}, \bibinfo {author} {\bibfnamefont {S.}~\bibnamefont {Hadfield}}, \bibinfo {author} {\bibfnamefont {M.~S.}\ \bibnamefont {Alam}}, \bibinfo {author} {\bibfnamefont {Z.}~\bibnamefont {Wang}}, \bibinfo {author} {\bibfnamefont {S.}~\bibnamefont {Grabbe}}, \bibinfo {author} {\bibfnamefont {P.~A.}\ \bibnamefont {Lott}}, \bibinfo {author} {\bibfnamefont {E.~G.}\ \bibnamefont {Rieffel}}, \bibinfo {author} {\bibfnamefont {D.}~\bibnamefont {Venturelli}},\ and\ \bibinfo {author} {\bibfnamefont
  {M.~J.}\ \bibnamefont {Reagor}},\ }\bibfield  {title} {\bibinfo {title} {Quantum-enhanced greedy combinatorial optimization solver},\ }\href {https://doi.org/10.1126/sciadv.adi0487} {\bibfield  {journal} {\bibinfo  {journal} {Science Advances}\ }\textbf {\bibinfo {volume} {9}},\ \bibinfo {pages} {eadi0487} (\bibinfo {year} {2023})},\ \Eprint {https://arxiv.org/abs/https://www.science.org/doi/pdf/10.1126/sciadv.adi0487} {https://www.science.org/doi/pdf/10.1126/sciadv.adi0487} \BibitemShut {NoStop}%
\bibitem [{\citenamefont {Het{\'e}nyi}\ and\ \citenamefont {Wootton}(2024{\natexlab{a}})}]{hetenyi2024tailoring}%
  \BibitemOpen
  \bibfield  {author} {\bibinfo {author} {\bibfnamefont {B.}~\bibnamefont {Het{\'e}nyi}}\ and\ \bibinfo {author} {\bibfnamefont {J.~R.}\ \bibnamefont {Wootton}},\ }\bibfield  {title} {\bibinfo {title} {Tailoring quantum error correction to spin qubits},\ }\href@noop {} {\bibfield  {journal} {\bibinfo  {journal} {Physical Review A}\ }\textbf {\bibinfo {volume} {109}},\ \bibinfo {pages} {032433} (\bibinfo {year} {2024}{\natexlab{a}})}\BibitemShut {NoStop}%
\bibitem [{\citenamefont {Wu}\ \emph {et~al.}(2021)\citenamefont {Wu}, \citenamefont {Li}, \citenamefont {Zhang}, \citenamefont {Guerreschi}, \citenamefont {Ding},\ and\ \citenamefont {Xie}}]{wu2021mapping}%
  \BibitemOpen
  \bibfield  {author} {\bibinfo {author} {\bibfnamefont {A.}~\bibnamefont {Wu}}, \bibinfo {author} {\bibfnamefont {G.}~\bibnamefont {Li}}, \bibinfo {author} {\bibfnamefont {H.}~\bibnamefont {Zhang}}, \bibinfo {author} {\bibfnamefont {G.~G.}\ \bibnamefont {Guerreschi}}, \bibinfo {author} {\bibfnamefont {Y.}~\bibnamefont {Ding}},\ and\ \bibinfo {author} {\bibfnamefont {Y.}~\bibnamefont {Xie}},\ }\bibfield  {title} {\bibinfo {title} {Mapping surface code to superconducting quantum processors},\ }\href@noop {} {\bibfield  {journal} {\bibinfo  {journal} {arXiv preprint arXiv:2111.13729}\ } (\bibinfo {year} {2021})}\BibitemShut {NoStop}%
\bibitem [{\citenamefont {Kremer}\ \emph {et~al.}(2024)\citenamefont {Kremer}, \citenamefont {Villar}, \citenamefont {Paik}, \citenamefont {Duran}, \citenamefont {Faro},\ and\ \citenamefont {Cruz-Benito}}]{kremer2024practical}%
  \BibitemOpen
  \bibfield  {author} {\bibinfo {author} {\bibfnamefont {D.}~\bibnamefont {Kremer}}, \bibinfo {author} {\bibfnamefont {V.}~\bibnamefont {Villar}}, \bibinfo {author} {\bibfnamefont {H.}~\bibnamefont {Paik}}, \bibinfo {author} {\bibfnamefont {I.}~\bibnamefont {Duran}}, \bibinfo {author} {\bibfnamefont {I.}~\bibnamefont {Faro}},\ and\ \bibinfo {author} {\bibfnamefont {J.}~\bibnamefont {Cruz-Benito}},\ }\href@noop {} {\bibinfo {title} {Practical and efficient quantum circuit synthesis and transpiling with reinforcement learning}} (\bibinfo {year} {2024}),\ \Eprint {https://arxiv.org/abs/2405.13196} {arXiv:2405.13196 [quant-ph]} \BibitemShut {NoStop}%
\bibitem [{\citenamefont {Siraichi}\ \emph {et~al.}(2018)\citenamefont {Siraichi}, \citenamefont {Santos}, \citenamefont {Collange},\ and\ \citenamefont {Pereira}}]{Siraichi2018}%
  \BibitemOpen
  \bibfield  {author} {\bibinfo {author} {\bibfnamefont {M.~Y.}\ \bibnamefont {Siraichi}}, \bibinfo {author} {\bibfnamefont {V.~F.~d.}\ \bibnamefont {Santos}}, \bibinfo {author} {\bibfnamefont {C.}~\bibnamefont {Collange}},\ and\ \bibinfo {author} {\bibfnamefont {F.~M.~Q.}\ \bibnamefont {Pereira}},\ }\bibfield  {title} {\bibinfo {title} {Qubit allocation},\ }in\ \href {https://doi.org/10.1145/3168822} {\emph {\bibinfo {booktitle} {Proceedings of the 2018 International Symposium on Code Generation and Optimization}}},\ \bibinfo {series and number} {CGO 2018}\ (\bibinfo  {publisher} {Association for Computing Machinery},\ \bibinfo {address} {New York, NY, USA},\ \bibinfo {year} {2018})\ p.\ \bibinfo {pages} {113–125}\BibitemShut {NoStop}%
\bibitem [{\citenamefont {Choe}\ and\ \citenamefont {Koenig}(2024)}]{choe2024fault}%
  \BibitemOpen
  \bibfield  {author} {\bibinfo {author} {\bibfnamefont {S.~H.}\ \bibnamefont {Choe}}\ and\ \bibinfo {author} {\bibfnamefont {R.}~\bibnamefont {Koenig}},\ }\bibfield  {title} {\bibinfo {title} {How to fault-tolerantly realize any quantum circuit with local operations},\ }\href@noop {} {\bibfield  {journal} {\bibinfo  {journal} {arXiv preprint arXiv:2402.13863}\ } (\bibinfo {year} {2024})}\BibitemShut {NoStop}%
\bibitem [{\citenamefont {Pino}\ \emph {et~al.}(2021)\citenamefont {Pino}, \citenamefont {Dreiling}, \citenamefont {Figgatt}, \citenamefont {Gaebler}, \citenamefont {Moses}, \citenamefont {Allman}, \citenamefont {Baldwin}, \citenamefont {Foss-Feig}, \citenamefont {Hayes}, \citenamefont {Mayer}, \citenamefont {Ryan-Anderson},\ and\ \citenamefont {Neyenhuis}}]{pino2021demonstration}%
  \BibitemOpen
  \bibfield  {author} {\bibinfo {author} {\bibfnamefont {J.~M.}\ \bibnamefont {Pino}}, \bibinfo {author} {\bibfnamefont {J.~M.}\ \bibnamefont {Dreiling}}, \bibinfo {author} {\bibfnamefont {C.}~\bibnamefont {Figgatt}}, \bibinfo {author} {\bibfnamefont {J.~P.}\ \bibnamefont {Gaebler}}, \bibinfo {author} {\bibfnamefont {S.~A.}\ \bibnamefont {Moses}}, \bibinfo {author} {\bibfnamefont {M.~S.}\ \bibnamefont {Allman}}, \bibinfo {author} {\bibfnamefont {C.~H.}\ \bibnamefont {Baldwin}}, \bibinfo {author} {\bibfnamefont {M.}~\bibnamefont {Foss-Feig}}, \bibinfo {author} {\bibfnamefont {D.}~\bibnamefont {Hayes}}, \bibinfo {author} {\bibfnamefont {K.}~\bibnamefont {Mayer}}, \bibinfo {author} {\bibfnamefont {C.}~\bibnamefont {Ryan-Anderson}},\ and\ \bibinfo {author} {\bibfnamefont {B.}~\bibnamefont {Neyenhuis}},\ }\bibfield  {title} {\bibinfo {title} {Demonstration of the trapped-ion quantum ccd computer architecture},\ }\href {https://doi.org/10.1038/s41586-021-03318-4} {\bibfield  {journal} {\bibinfo  {journal} {Nature}\
  }\textbf {\bibinfo {volume} {592}},\ \bibinfo {pages} {209–213} (\bibinfo {year} {2021})}\BibitemShut {NoStop}%
\bibitem [{\citenamefont {Bluvstein}\ \emph {et~al.}(2023)\citenamefont {Bluvstein}, \citenamefont {Evered}, \citenamefont {Geim}, \citenamefont {Li}, \citenamefont {Zhou}, \citenamefont {Manovitz}, \citenamefont {Ebadi}, \citenamefont {Cain}, \citenamefont {Kalinowski}, \citenamefont {Hangleiter}, \citenamefont {Bonilla~Ataides}, \citenamefont {Maskara}, \citenamefont {Cong}, \citenamefont {Gao}, \citenamefont {Sales~Rodriguez}, \citenamefont {Karolyshyn}, \citenamefont {Semeghini}, \citenamefont {Gullans}, \citenamefont {Greiner}, \citenamefont {Vuletić},\ and\ \citenamefont {Lukin}}]{Bluvstein_2023}%
  \BibitemOpen
  \bibfield  {author} {\bibinfo {author} {\bibfnamefont {D.}~\bibnamefont {Bluvstein}}, \bibinfo {author} {\bibfnamefont {S.~J.}\ \bibnamefont {Evered}}, \bibinfo {author} {\bibfnamefont {A.~A.}\ \bibnamefont {Geim}}, \bibinfo {author} {\bibfnamefont {S.~H.}\ \bibnamefont {Li}}, \bibinfo {author} {\bibfnamefont {H.}~\bibnamefont {Zhou}}, \bibinfo {author} {\bibfnamefont {T.}~\bibnamefont {Manovitz}}, \bibinfo {author} {\bibfnamefont {S.}~\bibnamefont {Ebadi}}, \bibinfo {author} {\bibfnamefont {M.}~\bibnamefont {Cain}}, \bibinfo {author} {\bibfnamefont {M.}~\bibnamefont {Kalinowski}}, \bibinfo {author} {\bibfnamefont {D.}~\bibnamefont {Hangleiter}}, \bibinfo {author} {\bibfnamefont {J.~P.}\ \bibnamefont {Bonilla~Ataides}}, \bibinfo {author} {\bibfnamefont {N.}~\bibnamefont {Maskara}}, \bibinfo {author} {\bibfnamefont {I.}~\bibnamefont {Cong}}, \bibinfo {author} {\bibfnamefont {X.}~\bibnamefont {Gao}}, \bibinfo {author} {\bibfnamefont {P.}~\bibnamefont {Sales~Rodriguez}}, \bibinfo {author} {\bibfnamefont
  {T.}~\bibnamefont {Karolyshyn}}, \bibinfo {author} {\bibfnamefont {G.}~\bibnamefont {Semeghini}}, \bibinfo {author} {\bibfnamefont {M.~J.}\ \bibnamefont {Gullans}}, \bibinfo {author} {\bibfnamefont {M.}~\bibnamefont {Greiner}}, \bibinfo {author} {\bibfnamefont {V.}~\bibnamefont {Vuletić}},\ and\ \bibinfo {author} {\bibfnamefont {M.~D.}\ \bibnamefont {Lukin}},\ }\bibfield  {title} {\bibinfo {title} {Logical quantum processor based on reconfigurable atom arrays},\ }\href {https://doi.org/10.1038/s41586-023-06927-3} {\bibfield  {journal} {\bibinfo  {journal} {Nature}\ }\textbf {\bibinfo {volume} {626}},\ \bibinfo {pages} {58–65} (\bibinfo {year} {2023})}\BibitemShut {NoStop}%
\bibitem [{\citenamefont {Stade}\ \emph {et~al.}(2024)\citenamefont {Stade}, \citenamefont {Schmid}, \citenamefont {Burgholzer},\ and\ \citenamefont {Wille}}]{stade2024abstract}%
  \BibitemOpen
  \bibfield  {author} {\bibinfo {author} {\bibfnamefont {Y.}~\bibnamefont {Stade}}, \bibinfo {author} {\bibfnamefont {L.}~\bibnamefont {Schmid}}, \bibinfo {author} {\bibfnamefont {L.}~\bibnamefont {Burgholzer}},\ and\ \bibinfo {author} {\bibfnamefont {R.}~\bibnamefont {Wille}},\ }\bibfield  {title} {\bibinfo {title} {An abstract model and efficient routing for logical entangling gates on zoned neutral atom architectures},\ }\href@noop {} {\bibfield  {journal} {\bibinfo  {journal} {arXiv preprint arXiv:2405.08068}\ } (\bibinfo {year} {2024})}\BibitemShut {NoStop}%
\bibitem [{\citenamefont {Finigan}\ \emph {et~al.}(2018)\citenamefont {Finigan}, \citenamefont {Cubeddu}, \citenamefont {Lively}, \citenamefont {Flick},\ and\ \citenamefont {Narang}}]{finigan2018qubit}%
  \BibitemOpen
  \bibfield  {author} {\bibinfo {author} {\bibfnamefont {W.}~\bibnamefont {Finigan}}, \bibinfo {author} {\bibfnamefont {M.}~\bibnamefont {Cubeddu}}, \bibinfo {author} {\bibfnamefont {T.}~\bibnamefont {Lively}}, \bibinfo {author} {\bibfnamefont {J.}~\bibnamefont {Flick}},\ and\ \bibinfo {author} {\bibfnamefont {P.}~\bibnamefont {Narang}},\ }\href@noop {} {\bibinfo {title} {Qubit allocation for noisy intermediate-scale quantum computers}} (\bibinfo {year} {2018}),\ \Eprint {https://arxiv.org/abs/1810.08291} {arXiv:1810.08291 [quant-ph]} \BibitemShut {NoStop}%
\bibitem [{\citenamefont {Li}\ \emph {et~al.}(2019)\citenamefont {Li}, \citenamefont {Ding},\ and\ \citenamefont {Xie}}]{li2019tackling}%
  \BibitemOpen
  \bibfield  {author} {\bibinfo {author} {\bibfnamefont {G.}~\bibnamefont {Li}}, \bibinfo {author} {\bibfnamefont {Y.}~\bibnamefont {Ding}},\ and\ \bibinfo {author} {\bibfnamefont {Y.}~\bibnamefont {Xie}},\ }\bibfield  {title} {\bibinfo {title} {Tackling the qubit mapping problem for nisq-era quantum devices},\ }in\ \href@noop {} {\emph {\bibinfo {booktitle} {Proceedings of the Twenty-Fourth International Conference on Architectural Support for Programming Languages and Operating Systems}}}\ (\bibinfo {year} {2019})\ pp.\ \bibinfo {pages} {1001--1014}\BibitemShut {NoStop}%
\bibitem [{\citenamefont {Childs}\ \emph {et~al.}(2019)\citenamefont {Childs}, \citenamefont {Schoute},\ and\ \citenamefont {Unsal}}]{childs2019circuit}%
  \BibitemOpen
  \bibfield  {author} {\bibinfo {author} {\bibfnamefont {A.~M.}\ \bibnamefont {Childs}}, \bibinfo {author} {\bibfnamefont {E.}~\bibnamefont {Schoute}},\ and\ \bibinfo {author} {\bibfnamefont {C.~M.}\ \bibnamefont {Unsal}},\ }\bibfield  {title} {\bibinfo {title} {Circuit transformations for quantum architectures},\ }\href@noop {} {\bibfield  {journal} {\bibinfo  {journal} {arXiv preprint arXiv:1902.09102}\ } (\bibinfo {year} {2019})}\BibitemShut {NoStop}%
\bibitem [{\citenamefont {Nannicini}\ \emph {et~al.}(2022)\citenamefont {Nannicini}, \citenamefont {Bishop}, \citenamefont {G\"{u}nl\"{u}k},\ and\ \citenamefont {Jurcevic}}]{nannicini2022optimal}%
  \BibitemOpen
  \bibfield  {author} {\bibinfo {author} {\bibfnamefont {G.}~\bibnamefont {Nannicini}}, \bibinfo {author} {\bibfnamefont {L.~S.}\ \bibnamefont {Bishop}}, \bibinfo {author} {\bibfnamefont {O.}~\bibnamefont {G\"{u}nl\"{u}k}},\ and\ \bibinfo {author} {\bibfnamefont {P.}~\bibnamefont {Jurcevic}},\ }\bibfield  {title} {\bibinfo {title} {Optimal qubit assignment and routing via integer programming},\ }\bibfield  {journal} {\bibinfo  {journal} {ACM Transactions on Quantum Computing}\ }\textbf {\bibinfo {volume} {4}},\ \href {https://doi.org/10.1145/3544563} {10.1145/3544563} (\bibinfo {year} {2022})\BibitemShut {NoStop}%
\bibitem [{\citenamefont {Ito}\ \emph {et~al.}(2023)\citenamefont {Ito}, \citenamefont {Kakimura}, \citenamefont {Kamiyama}, \citenamefont {Kobayashi},\ and\ \citenamefont {Okamoto}}]{ito2023algorithmic}%
  \BibitemOpen
  \bibfield  {author} {\bibinfo {author} {\bibfnamefont {T.}~\bibnamefont {Ito}}, \bibinfo {author} {\bibfnamefont {N.}~\bibnamefont {Kakimura}}, \bibinfo {author} {\bibfnamefont {N.}~\bibnamefont {Kamiyama}}, \bibinfo {author} {\bibfnamefont {Y.}~\bibnamefont {Kobayashi}},\ and\ \bibinfo {author} {\bibfnamefont {Y.}~\bibnamefont {Okamoto}},\ }\href@noop {} {\bibinfo {title} {Algorithmic theory of qubit routing}} (\bibinfo {year} {2023}),\ \Eprint {https://arxiv.org/abs/2305.02059} {arXiv:2305.02059 [cs.DS]} \BibitemShut {NoStop}%
\bibitem [{\citenamefont {Gottesman}(2022)}]{gottesman2022opportunities}%
  \BibitemOpen
  \bibfield  {author} {\bibinfo {author} {\bibfnamefont {D.}~\bibnamefont {Gottesman}},\ }\href@noop {} {\bibinfo {title} {Opportunities and challenges in fault-tolerant quantum computation}} (\bibinfo {year} {2022}),\ \Eprint {https://arxiv.org/abs/2210.15844} {arXiv:2210.15844 [quant-ph]} \BibitemShut {NoStop}%
\bibitem [{\citenamefont {Aliferis}\ and\ \citenamefont {Preskill}(2008)}]{AP2008}%
  \BibitemOpen
  \bibfield  {author} {\bibinfo {author} {\bibfnamefont {P.}~\bibnamefont {Aliferis}}\ and\ \bibinfo {author} {\bibfnamefont {J.}~\bibnamefont {Preskill}},\ }\bibfield  {title} {\bibinfo {title} {Fault-tolerant quantum computation against biased noise},\ }\href {https://doi.org/10.1103/PhysRevA.78.052331} {\bibfield  {journal} {\bibinfo  {journal} {Phys. Rev. A}\ }\textbf {\bibinfo {volume} {78}},\ \bibinfo {pages} {052331} (\bibinfo {year} {2008})}\BibitemShut {NoStop}%
\bibitem [{\citenamefont {Tuckett}\ \emph {et~al.}(2020)\citenamefont {Tuckett}, \citenamefont {Bartlett}, \citenamefont {Flammia},\ and\ \citenamefont {Brown}}]{Tuckett2020}%
  \BibitemOpen
  \bibfield  {author} {\bibinfo {author} {\bibfnamefont {D.~K.}\ \bibnamefont {Tuckett}}, \bibinfo {author} {\bibfnamefont {S.~D.}\ \bibnamefont {Bartlett}}, \bibinfo {author} {\bibfnamefont {S.~T.}\ \bibnamefont {Flammia}},\ and\ \bibinfo {author} {\bibfnamefont {B.~J.}\ \bibnamefont {Brown}},\ }\bibfield  {title} {\bibinfo {title} {Fault-tolerant thresholds for the surface code in excess of 5\% under biased noise},\ }\href {https://doi.org/10.1103/PhysRevLett.124.130501} {\bibfield  {journal} {\bibinfo  {journal} {Phys. Rev. Lett.}\ }\textbf {\bibinfo {volume} {124}},\ \bibinfo {pages} {130501} (\bibinfo {year} {2020})}\BibitemShut {NoStop}%
\bibitem [{\citenamefont {Gottesman}(2013)}]{gottesman2013fault}%
  \BibitemOpen
  \bibfield  {author} {\bibinfo {author} {\bibfnamefont {D.}~\bibnamefont {Gottesman}},\ }\bibfield  {title} {\bibinfo {title} {Fault-tolerant quantum computation with constant overhead},\ }\href@noop {} {\bibfield  {journal} {\bibinfo  {journal} {arXiv preprint arXiv:1310.2984}\ } (\bibinfo {year} {2013})}\BibitemShut {NoStop}%
\bibitem [{\citenamefont {Delfosse}\ and\ \citenamefont {Reichardt}(2020)}]{delfosse2020short}%
  \BibitemOpen
  \bibfield  {author} {\bibinfo {author} {\bibfnamefont {N.}~\bibnamefont {Delfosse}}\ and\ \bibinfo {author} {\bibfnamefont {B.~W.}\ \bibnamefont {Reichardt}},\ }\bibfield  {title} {\bibinfo {title} {Short shor-style syndrome sequences},\ }\href@noop {} {\bibfield  {journal} {\bibinfo  {journal} {arXiv preprint arXiv:2008.05051}\ } (\bibinfo {year} {2020})}\BibitemShut {NoStop}%
\bibitem [{\citenamefont {Tansuwannont}\ \emph {et~al.}(2023)\citenamefont {Tansuwannont}, \citenamefont {Pato},\ and\ \citenamefont {Brown}}]{tansuwannont2023adaptive}%
  \BibitemOpen
  \bibfield  {author} {\bibinfo {author} {\bibfnamefont {T.}~\bibnamefont {Tansuwannont}}, \bibinfo {author} {\bibfnamefont {B.}~\bibnamefont {Pato}},\ and\ \bibinfo {author} {\bibfnamefont {K.~R.}\ \bibnamefont {Brown}},\ }\bibfield  {title} {\bibinfo {title} {Adaptive syndrome measurements for shor-style error correction},\ }\href@noop {} {\bibfield  {journal} {\bibinfo  {journal} {Quantum}\ }\textbf {\bibinfo {volume} {7}},\ \bibinfo {pages} {1075} (\bibinfo {year} {2023})}\BibitemShut {NoStop}%
\bibitem [{\citenamefont {Fowler}\ \emph {et~al.}(2012)\citenamefont {Fowler}, \citenamefont {Mariantoni}, \citenamefont {Martinis},\ and\ \citenamefont {Cleland}}]{fowler2012surface}%
  \BibitemOpen
  \bibfield  {author} {\bibinfo {author} {\bibfnamefont {A.~G.}\ \bibnamefont {Fowler}}, \bibinfo {author} {\bibfnamefont {M.}~\bibnamefont {Mariantoni}}, \bibinfo {author} {\bibfnamefont {J.~M.}\ \bibnamefont {Martinis}},\ and\ \bibinfo {author} {\bibfnamefont {A.~N.}\ \bibnamefont {Cleland}},\ }\bibfield  {title} {\bibinfo {title} {Surface codes: Towards practical large-scale quantum computation},\ }\href@noop {} {\bibfield  {journal} {\bibinfo  {journal} {Physical Review A}\ }\textbf {\bibinfo {volume} {86}},\ \bibinfo {pages} {032324} (\bibinfo {year} {2012})}\BibitemShut {NoStop}%
\bibitem [{\citenamefont {Tomita}\ and\ \citenamefont {Svore}(2014)}]{tomita2014low}%
  \BibitemOpen
  \bibfield  {author} {\bibinfo {author} {\bibfnamefont {Y.}~\bibnamefont {Tomita}}\ and\ \bibinfo {author} {\bibfnamefont {K.~M.}\ \bibnamefont {Svore}},\ }\bibfield  {title} {\bibinfo {title} {Low-distance surface codes under realistic quantum noise},\ }\href@noop {} {\bibfield  {journal} {\bibinfo  {journal} {Physical Review A}\ }\textbf {\bibinfo {volume} {90}},\ \bibinfo {pages} {062320} (\bibinfo {year} {2014})}\BibitemShut {NoStop}%
\bibitem [{\citenamefont {Gidney}(2021)}]{gidney2021stim}%
  \BibitemOpen
  \bibfield  {author} {\bibinfo {author} {\bibfnamefont {C.}~\bibnamefont {Gidney}},\ }\bibfield  {title} {\bibinfo {title} {Stim: a fast stabilizer circuit simulator},\ }\href@noop {} {\bibfield  {journal} {\bibinfo  {journal} {Quantum}\ }\textbf {\bibinfo {volume} {5}},\ \bibinfo {pages} {497} (\bibinfo {year} {2021})}\BibitemShut {NoStop}%
\bibitem [{\citenamefont {Higgott}(2022)}]{higgott2022pymatching}%
  \BibitemOpen
  \bibfield  {author} {\bibinfo {author} {\bibfnamefont {O.}~\bibnamefont {Higgott}},\ }\bibfield  {title} {\bibinfo {title} {Pymatching: A python package for decoding quantum codes with minimum-weight perfect matching},\ }\href@noop {} {\bibfield  {journal} {\bibinfo  {journal} {ACM Transactions on Quantum Computing}\ }\textbf {\bibinfo {volume} {3}},\ \bibinfo {pages} {1} (\bibinfo {year} {2022})}\BibitemShut {NoStop}%
\bibitem [{\citenamefont {Higgott}\ and\ \citenamefont {Gidney}(2023)}]{higgott2023sparse}%
  \BibitemOpen
  \bibfield  {author} {\bibinfo {author} {\bibfnamefont {O.}~\bibnamefont {Higgott}}\ and\ \bibinfo {author} {\bibfnamefont {C.}~\bibnamefont {Gidney}},\ }\bibfield  {title} {\bibinfo {title} {Sparse blossom: correcting a million errors per core second with minimum-weight matching},\ }\href@noop {} {\bibfield  {journal} {\bibinfo  {journal} {arXiv preprint arXiv:2303.15933}\ } (\bibinfo {year} {2023})}\BibitemShut {NoStop}%
\bibitem [{\citenamefont {McEwen}\ \emph {et~al.}(2023)\citenamefont {McEwen}, \citenamefont {Bacon},\ and\ \citenamefont {Gidney}}]{mcewen2023relaxing}%
  \BibitemOpen
  \bibfield  {author} {\bibinfo {author} {\bibfnamefont {M.}~\bibnamefont {McEwen}}, \bibinfo {author} {\bibfnamefont {D.}~\bibnamefont {Bacon}},\ and\ \bibinfo {author} {\bibfnamefont {C.}~\bibnamefont {Gidney}},\ }\bibfield  {title} {\bibinfo {title} {Relaxing hardware requirements for surface code circuits using time-dynamics},\ }\href@noop {} {\bibfield  {journal} {\bibinfo  {journal} {arXiv preprint arXiv:2302.02192}\ } (\bibinfo {year} {2023})}\BibitemShut {NoStop}%
\bibitem [{\citenamefont {Benito}\ \emph {et~al.}(2024)\citenamefont {Benito}, \citenamefont {L{\'o}pez}, \citenamefont {Peropadre},\ and\ \citenamefont {Bermudez}}]{benito2024comparative}%
  \BibitemOpen
  \bibfield  {author} {\bibinfo {author} {\bibfnamefont {C.}~\bibnamefont {Benito}}, \bibinfo {author} {\bibfnamefont {E.}~\bibnamefont {L{\'o}pez}}, \bibinfo {author} {\bibfnamefont {B.}~\bibnamefont {Peropadre}},\ and\ \bibinfo {author} {\bibfnamefont {A.}~\bibnamefont {Bermudez}},\ }\bibfield  {title} {\bibinfo {title} {Comparative study of quantum error correction strategies for the heavy-hexagonal lattice},\ }\href@noop {} {\bibfield  {journal} {\bibinfo  {journal} {arXiv preprint arXiv:2402.02185}\ } (\bibinfo {year} {2024})}\BibitemShut {NoStop}%
\bibitem [{\citenamefont {Het{\'e}nyi}\ and\ \citenamefont {Wootton}(2024{\natexlab{b}})}]{hetenyi2024creating}%
  \BibitemOpen
  \bibfield  {author} {\bibinfo {author} {\bibfnamefont {B.}~\bibnamefont {Het{\'e}nyi}}\ and\ \bibinfo {author} {\bibfnamefont {J.~R.}\ \bibnamefont {Wootton}},\ }\bibfield  {title} {\bibinfo {title} {Creating entangled logical qubits in the heavy-hex lattice with topological codes},\ }\href@noop {} {\bibfield  {journal} {\bibinfo  {journal} {arXiv preprint arXiv:2404.15989}\ } (\bibinfo {year} {2024}{\natexlab{b}})}\BibitemShut {NoStop}%
\bibitem [{\citenamefont {Terhal}\ and\ \citenamefont {Burkard}(2005)}]{terhal2005fault}%
  \BibitemOpen
  \bibfield  {author} {\bibinfo {author} {\bibfnamefont {B.~M.}\ \bibnamefont {Terhal}}\ and\ \bibinfo {author} {\bibfnamefont {G.}~\bibnamefont {Burkard}},\ }\bibfield  {title} {\bibinfo {title} {Fault-tolerant quantum computation for local non-markovian noise},\ }\href@noop {} {\bibfield  {journal} {\bibinfo  {journal} {Physical Review A}\ }\textbf {\bibinfo {volume} {71}},\ \bibinfo {pages} {012336} (\bibinfo {year} {2005})}\BibitemShut {NoStop}%
\bibitem [{\citenamefont {Aharonov}\ \emph {et~al.}(2006)\citenamefont {Aharonov}, \citenamefont {Kitaev},\ and\ \citenamefont {Preskill}}]{aharonov2006fault}%
  \BibitemOpen
  \bibfield  {author} {\bibinfo {author} {\bibfnamefont {D.}~\bibnamefont {Aharonov}}, \bibinfo {author} {\bibfnamefont {A.}~\bibnamefont {Kitaev}},\ and\ \bibinfo {author} {\bibfnamefont {J.}~\bibnamefont {Preskill}},\ }\bibfield  {title} {\bibinfo {title} {Fault-tolerant quantum computation with long-range correlated noise},\ }\href@noop {} {\bibfield  {journal} {\bibinfo  {journal} {Physical review letters}\ }\textbf {\bibinfo {volume} {96}},\ \bibinfo {pages} {050504} (\bibinfo {year} {2006})}\BibitemShut {NoStop}%
\bibitem [{\citenamefont {Ng}\ and\ \citenamefont {Preskill}(2009)}]{ng2009fault}%
  \BibitemOpen
  \bibfield  {author} {\bibinfo {author} {\bibfnamefont {H.~K.}\ \bibnamefont {Ng}}\ and\ \bibinfo {author} {\bibfnamefont {J.}~\bibnamefont {Preskill}},\ }\bibfield  {title} {\bibinfo {title} {Fault-tolerant quantum computation versus gaussian noise},\ }\href@noop {} {\bibfield  {journal} {\bibinfo  {journal} {Physical Review A}\ }\textbf {\bibinfo {volume} {79}},\ \bibinfo {pages} {032318} (\bibinfo {year} {2009})}\BibitemShut {NoStop}%
\bibitem [{\citenamefont {Higgott}\ \emph {et~al.}(2023)\citenamefont {Higgott}, \citenamefont {Bohdanowicz}, \citenamefont {Kubica}, \citenamefont {Flammia},\ and\ \citenamefont {Campbell}}]{higgott2023improved}%
  \BibitemOpen
  \bibfield  {author} {\bibinfo {author} {\bibfnamefont {O.}~\bibnamefont {Higgott}}, \bibinfo {author} {\bibfnamefont {T.~C.}\ \bibnamefont {Bohdanowicz}}, \bibinfo {author} {\bibfnamefont {A.}~\bibnamefont {Kubica}}, \bibinfo {author} {\bibfnamefont {S.~T.}\ \bibnamefont {Flammia}},\ and\ \bibinfo {author} {\bibfnamefont {E.~T.}\ \bibnamefont {Campbell}},\ }\bibfield  {title} {\bibinfo {title} {Improved decoding of circuit noise and fragile boundaries of tailored surface codes},\ }\href@noop {} {\bibfield  {journal} {\bibinfo  {journal} {Physical Review X}\ }\textbf {\bibinfo {volume} {13}},\ \bibinfo {pages} {031007} (\bibinfo {year} {2023})}\BibitemShut {NoStop}%
\bibitem [{\citenamefont {Chen}\ \emph {et~al.}(2022)\citenamefont {Chen}, \citenamefont {Yoder}, \citenamefont {Kim}, \citenamefont {Sundaresan}, \citenamefont {Srinivasan}, \citenamefont {Li}, \citenamefont {C{\'o}rcoles}, \citenamefont {Cross},\ and\ \citenamefont {Takita}}]{chen2022calibrated}%
  \BibitemOpen
  \bibfield  {author} {\bibinfo {author} {\bibfnamefont {E.~H.}\ \bibnamefont {Chen}}, \bibinfo {author} {\bibfnamefont {T.~J.}\ \bibnamefont {Yoder}}, \bibinfo {author} {\bibfnamefont {Y.}~\bibnamefont {Kim}}, \bibinfo {author} {\bibfnamefont {N.}~\bibnamefont {Sundaresan}}, \bibinfo {author} {\bibfnamefont {S.}~\bibnamefont {Srinivasan}}, \bibinfo {author} {\bibfnamefont {M.}~\bibnamefont {Li}}, \bibinfo {author} {\bibfnamefont {A.~D.}\ \bibnamefont {C{\'o}rcoles}}, \bibinfo {author} {\bibfnamefont {A.~W.}\ \bibnamefont {Cross}},\ and\ \bibinfo {author} {\bibfnamefont {M.}~\bibnamefont {Takita}},\ }\bibfield  {title} {\bibinfo {title} {Calibrated decoders for experimental quantum error correction},\ }\href@noop {} {\bibfield  {journal} {\bibinfo  {journal} {Physical Review Letters}\ }\textbf {\bibinfo {volume} {128}},\ \bibinfo {pages} {110504} (\bibinfo {year} {2022})}\BibitemShut {NoStop}%
\bibitem [{\citenamefont {Lin}\ \emph {et~al.}(2014)\citenamefont {Lin}, \citenamefont {Sur-Kolay},\ and\ \citenamefont {Jha}}]{lin2014paqcs}%
  \BibitemOpen
  \bibfield  {author} {\bibinfo {author} {\bibfnamefont {C.-C.}\ \bibnamefont {Lin}}, \bibinfo {author} {\bibfnamefont {S.}~\bibnamefont {Sur-Kolay}},\ and\ \bibinfo {author} {\bibfnamefont {N.~K.}\ \bibnamefont {Jha}},\ }\bibfield  {title} {\bibinfo {title} {Paqcs: Physical design-aware fault-tolerant quantum circuit synthesis},\ }\href@noop {} {\bibfield  {journal} {\bibinfo  {journal} {IEEE Transactions on Very Large Scale Integration (VLSI) Systems}\ }\textbf {\bibinfo {volume} {23}},\ \bibinfo {pages} {1221} (\bibinfo {year} {2014})}\BibitemShut {NoStop}%
\bibitem [{\citenamefont {Lao}\ \emph {et~al.}(2018)\citenamefont {Lao}, \citenamefont {van Wee}, \citenamefont {Ashraf}, \citenamefont {van Someren}, \citenamefont {Khammassi}, \citenamefont {Bertels},\ and\ \citenamefont {Almudever}}]{Lao_2019}%
  \BibitemOpen
  \bibfield  {author} {\bibinfo {author} {\bibfnamefont {L.}~\bibnamefont {Lao}}, \bibinfo {author} {\bibfnamefont {B.}~\bibnamefont {van Wee}}, \bibinfo {author} {\bibfnamefont {I.}~\bibnamefont {Ashraf}}, \bibinfo {author} {\bibfnamefont {J.}~\bibnamefont {van Someren}}, \bibinfo {author} {\bibfnamefont {N.}~\bibnamefont {Khammassi}}, \bibinfo {author} {\bibfnamefont {K.}~\bibnamefont {Bertels}},\ and\ \bibinfo {author} {\bibfnamefont {C.~G.}\ \bibnamefont {Almudever}},\ }\bibfield  {title} {\bibinfo {title} {Mapping of lattice surgery-based quantum circuits on surface code architectures},\ }\href {https://doi.org/10.1088/2058-9565/aadd1a} {\bibfield  {journal} {\bibinfo  {journal} {Quantum Science and Technology}\ }\textbf {\bibinfo {volume} {4}},\ \bibinfo {pages} {015005} (\bibinfo {year} {2018})}\BibitemShut {NoStop}%
\bibitem [{\citenamefont {Dennis}\ \emph {et~al.}(2002)\citenamefont {Dennis}, \citenamefont {Kitaev}, \citenamefont {Landahl},\ and\ \citenamefont {Preskill}}]{dennis2002topological}%
  \BibitemOpen
  \bibfield  {author} {\bibinfo {author} {\bibfnamefont {E.}~\bibnamefont {Dennis}}, \bibinfo {author} {\bibfnamefont {A.}~\bibnamefont {Kitaev}}, \bibinfo {author} {\bibfnamefont {A.}~\bibnamefont {Landahl}},\ and\ \bibinfo {author} {\bibfnamefont {J.}~\bibnamefont {Preskill}},\ }\bibfield  {title} {\bibinfo {title} {Topological quantum memory},\ }\href@noop {} {\bibfield  {journal} {\bibinfo  {journal} {Journal of Mathematical Physics}\ }\textbf {\bibinfo {volume} {43}},\ \bibinfo {pages} {4452} (\bibinfo {year} {2002})}\BibitemShut {NoStop}%
\end{thebibliography}%

%%%%%%%%%%%%%%%%%%%%%%%%%%
%%%%%%%%%%%%%%%%%%%%%%%%%%
\clearpage
\appendix

%%%%%%%%%%%%%%%%%%%%%%%%%%
%%%%%%%%%%%%%%%%%%%%%%%%%%
\section{Search algorithms for EPP routing schedules}\label{app:Routing}

Here, we discuss search algorithms for routing schedules and explain how to incorporate our EPP constraints in a straightforward manner. We focus on a simple search procedure based on a greedy, distance-minimizing search \cite{childs2019circuit}, but the same principles apply to other more elaborate methods \cite{lin2014paqcs,childs2019circuit,Lao_2019}.

We begin with an abstract quantum circuit given as an ordered product $U = U_LU_{L-1}\ldots U_1$ of $L$ layers, where each $U_i$ contains two-qubit gates with no overlapping supports. We initialize $E^1_I, ..., E^L_I$ to be the edges in the interaction graph $G_I$ corresponding to the two-qubit gates in $U_1, ..., U_L$, respectively. The interaction graph is denoted as $G_I\equiv(V_I, E_I)$, where $V_I(E_I)$ is the set of nodes(edges) of the graph; similarly, the device graph is $G_D\equiv (V_D,E_D)$. For the embedding to be possible, we must have $|V_D|\geq |V_I|$. A qubit mapping $\hat p_t$ denotes the collective assignment of abstract qubits to physical qubits at layer $t$.

Our goal here is to iteratively execute all two-qubit gates in $E^1_I, ..., E^L_I$, and we keep track of the executed gates by removing them from the set. This is achieved by using swap gates to bring abstract qubits next to one another in the device graph $G_D\equiv (V_D,E_D)$, so that the remaining gates can be executed. In other words, defining $d:V_D \times V_D$ as the geodesic or shortest-distance function between nodes in the device graph $G_D$, for all edges $(q_1, q_2) \in E^I_I$, swap gates are used to update the mapping $\hat{p}$ so that $d(q_1, q_2) = 1$, in which case the edge $(q_1, q_2)$ can be removed from $E^1_I$. To quantify the collective distance between qubits of all edges in an edgeset $E$, define the objective function:
\begin{equation} \label{eq:objective_func]}
    R(E) \equiv \sum_{(q_1, q_2) \in E} d(\hat{p}(q_1), \hat{p}(q_2)).
\end{equation}

Firstly, an initial embedding $\hat{p}_0$ is chosen, possibly in a manner that maximizes the number of executable edges in the first layer $E^1_I$. These edges are subsequently executed, i.e., removed from $E^1_I$. Next, consider the set of edges in $G_D$ where executing the corresponding swap gates at those edges maximally decreases $R(E^1_I)$. Here, we further impose the EPP constraints that allow only types-1 or 2 SWAPs. SWAPs are executed at these edges, the embedding is updated accordingly (to $\hat{p}_1$), and the set of remaining edges $E^1_I$ are checked to see if any satisfies $d = 1$, in which case they are executed (removed from $E^1_I$). If no SWAPs or gates in $E^1_I$ can be executed, one of the remaining edges in $E^1_I$ is selected arbitrarily, and a type-1 or 2 SWAP that decreases $d$ is executed. The above procedure is repeated until $E^1_I$ is empty, at which point we move on to edges in the second layer $E^2_I$, and so on until $E^1_I, ..., E^L_I$ are all empty. If, at any step of the way, a type-1 or 2 SWAP move cannot be found, the algorithm fails and terminates without a solution. 

If successful, the above algorithm returns an EPP routing schedule. If it terminates without a solution, either an EPP schedule simply does not exist---this may be very likely for highly constrained device graphs---or the above greedy search must be improved to search through more routing schedules, typically at the expense of higher computational costs. One such improvement is to increase the depth of the search in each iteration, i.e., allowing the algorithm to look more than one move ahead to avoid getting stuck at a local minima. We refer the reader to Ref.~\cite{childs2019circuit} for further discussion of this aspect.

We elaborate on how this search algorithm was modified to produce EPP routing schedules for topological codes such as the surface code in App.~\ref{app:SCRouting2}.

\section{Details on surface code examples}
Here, we provide additional details for the surface code examples discussed in the main text.

%%%%%%
\subsection{Search algorithm for surface code routing schedules} \label{app:SCRouting2}

Topological codes, including surface codes, possess a large degree of locality and translational symmetry in the defining code stabilizers, directly reflected in the SE circuits and the associated interaction graphs. In cases where the quantum device has similar symmetries (common in current quantum hardware architecture), these properties can be leveraged to reduce the size of the routing search problem. Here, we describe a search algorithm that yields constant-depth (i.e., does not scale with code distance $d$) routing schedules for the surface code error-correction circuit, with a computational cost that is independent of code distance. It can be straightforwardly generalized to other topological codes. This algorithm yields routing schedules for the examples discussed in the main text, and in Apps.~\ref{app:HeavyEx} and \ref{app:HexEx}.

Since the interaction graph of the surface code SE circuits and the device graphs are both periodic, with invariant local structures as the code-distance increases, the core idea behind our approach is to embed a unit cell of the interaction graph onto a unit cell of the device graph, and subsequently search for the routing schedule for this unit. This avoids the need to search across the entire interaction- and device-graph lattices, which increase in size with code distance. The resulting schedule can then be tiled in a repeated manner up to the desired code-distance, with tiles at the code boundary appropriately truncated to ensure that lower-weight stabilizers at the boundaries are measured correctly.

\begin{figure}
    \centerline{\includegraphics[width=.5\textwidth]{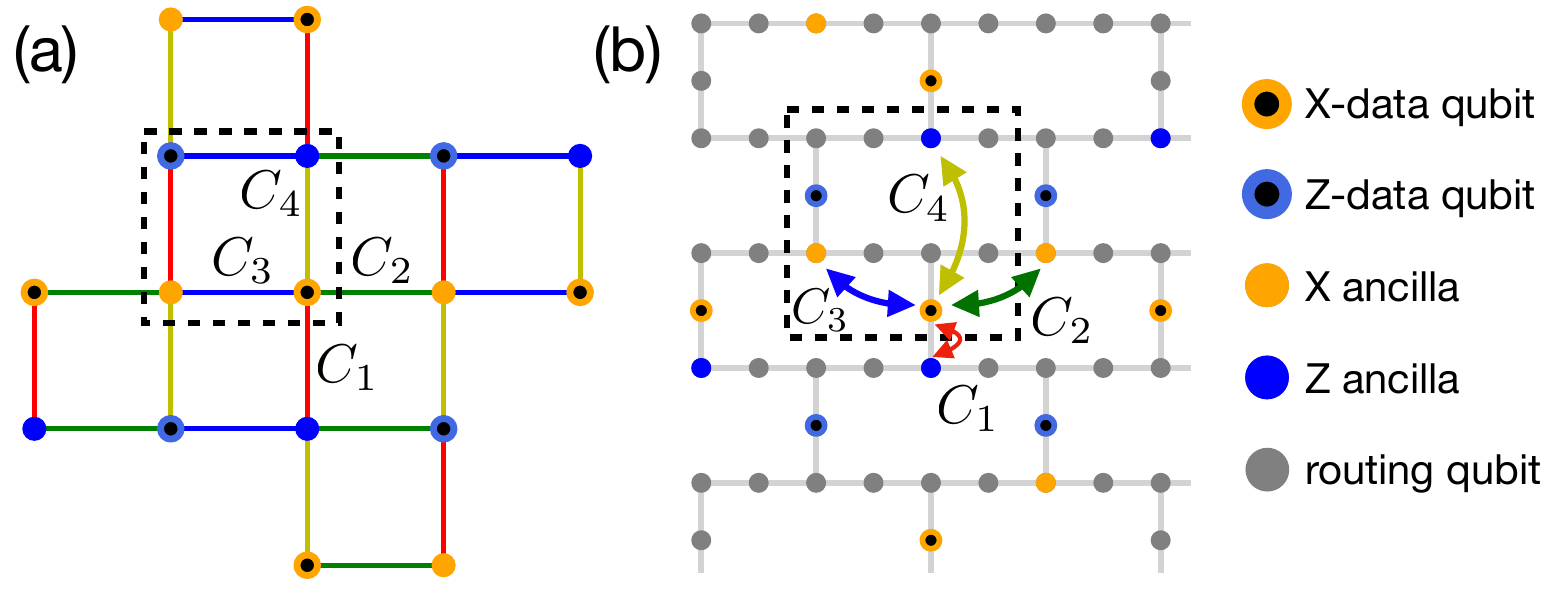}}
    \caption{The $d=3$ surface-code example. The interaction graph is given in (a), and the heavy-hexagonal device graph in (b). Unit cells are marked by the dotted boxes. The four distinct types of qubits of the surface code are labelled on the right. Also shown are the four CNOTs experienced by an $X$-data qubit; gates $C_1, C_2, C_3, C_4$ correspond to red, green, blue, yellow colored lines, respectively, in the interaction graph (a), while the correspondingly colored arrows in (b) mark the pairs of computational qubits in the device involved in these operations.
    }
    \label{fig:hhex_embedding_algo}
\end{figure}

We begin by identifying the units cells of the interaction and device graphs. A unit cell in the interaction graph of the surface code is illustrated in Fig.~\ref{fig:hhex_embedding_algo}(a). The computational qubits of the surface code can be classified into four species: the $Z$ ancillas, the $X$ ancillas, the $X$-data qubits (defined as a data qubit connected horizontally to two $X$ ancillas in that layer of the the interaction graph), and the $Z$-data qubits (similarly defined as a data qubit connected horizontally to two $Z$ ancillas). Each species of qubits can be distinguished by the ordering in which they interact with other qubits throughout the computation, which can be observed from the interaction graph and the ordering of CNOT gates of the surface code in Fig.~\ref{fig:SC}(a) and (b). For instance, while data qubits interact with $X$ and $Z$ ancillas in an alternating fashion, $X$- ($Z$-) data qubits interact with $Z$ ($X$) ancillas first. In the device graph, a unit cell that satisfies the same translational symmetry as the surface code is identified, marked by the dotted box in Fig.~\ref{fig:hhex_embedding_algo}(b). 

Next, a parametrization of the set of all routing schedules is required. Our approach consists of two parts, parametrizing the (1) initial embeddings and (2) swap layers separately:
\begin{enumerate}
\item \underline{Parametrization of initial embedding}. The four qubit species are assigned onto the unit cell of the device graph. If there are $N$ lattice sites on the unit cell, there are a total of $^N\!P_4$ possible assignments. For the heavy-hexagonal lattice, $N = 10$, giving $5040$ possible initial assignments.
\item \underline{Parametrization of swap layer}. A swap layer can be defined by specifying the targets of the swap gates applied to each of the four species of qubits. For the heavy-hexagonal lattice, each qubit can either be idling, or swapped with one of its two or three neighbors, so there are $\leq 4^4$ possible swap layers, which can be further reduced by enforcing the constraint on type-2 SWAPs. For convenience, we arrange the heavy-hexagonal and hexagonal lattices in a grid-like manner on a 2D plane as in Figs.~\ref{fig:SC}(b) and \ref{fig:hhex_embedding_algo}(b); this allows the action of each swap gate to be associated with a permutation with a qubit to its up (u), down (d), left (l), or right (r) directions, or idling (i).
\end{enumerate}
In summary, the number of possible $L$-swap-layer routing schedules for a QEC code with $n$ species of qubits embedded onto a device graph of maximum degree $\Delta$ with $N$ qubits per unit cell is upper-bounded by $(^N\!P_n \times n^{\Delta + 1})^L$.

The general search algorithm of App.~\ref{app:Routing} can then be adapted to obtain EPP routing schedules in a distance-independent manner:
\begin{enumerate}

    \item Initialize an embedding of the four species of qubits in one of the $^N\!P_4$ possible assignments in the unit cell in the device graph. For each of the four qubits, track the targets of the embedded CNOT gates (possibly located in adjacent unit cells), which defines the distance/objective function Eq.~(\ref{eq:objective_func]}). For example, for a particular $X$-data qubit, the targets of the effective gates are represented as curved colored lines labelled $C_1, C_2, C_3$, and $C_4$ in Fig.~\ref{fig:hhex_embedding_algo}(b), each corresponding to its abstract counterparts in the interaction graph of Fig.~\ref{fig:hhex_embedding_algo}(a).
    
    \item Apply a swap layer (consisting of one of the directions $u,d,l,r,i$ on each species of qubit) that minimizes the objective function Eq.~(\ref{eq:objective_func]}) for the CNOT layer $C_1$ in a greedy manner. The embedding is updated.
    \item Check for any executable gates in $C_1$, i.e., if the targets of the effective gates in layer $C_1$ are next to one another in the device graph in the current embedding.
    \item Repeat steps 2 and 3 until all gates in $C_1$ are executed, after which steps 2 and 3 are repeated sequentially with $C_2,C_3$, and $C_4$.
\end{enumerate}
\noindent If successful, the result is a set of swap layers $S_1, S_2, ...$ interspersed between the CNOT layers $C_1,C_2,C_3$, and $C_4$ that implements an EPP routing schedule.

%%%%%%%%%
\subsection{Effective noise model} 
\label{sec:approx_noise_model_surf}

Here, we describe how to derive an effective noise model on the embedded circuit for our surface-code examples. Apart from explaining how we arrive at the scaled error probability parameter $\peff$ used in our plots in the main text, this serves also as an explicit illustration of Lemma \ref{lemma} and Theorem \ref{theorem} for the specific example of Pauli noise.

The physical noise model used in our simulations is that of depolarizing noise: A single-qubit gate sees a Pauli $X$, $Y$, or $Z$ error with probability $p/3$ each; a two-qubit gate sees a two-qubit Pauli error with probability $p/15$ each. Idling qubits are assumed to be noiseless. Our discussion here, however, applies to a general Pauli noise model, where the single- and two-qubit gates see the following noise channels,
\begin{align}
\cE_{q}(\,\cdot\,)&\equiv \sum_\alpha p_\alpha\sigma_\alpha(\,\cdot\,)\sigma_\alpha\\
\cE_{q_1,q_2}(\,\cdot\,)&\equiv \sum_{\alpha\beta}p_{\alpha\beta}\sigma_\alpha\otimes\sigma_\beta(\,\cdot\,)\sigma_\alpha\otimes \sigma_\beta,\nonumber
\end{align}
where $\sigma_\alpha=I,X,Y,Z$, for $\alpha = 0,1,2,3$. Note the trace-preservation requirements: $p_0=1-\sum_{\alpha\neq 0}p_\alpha$ and $p_{00}=1-\sum_{\alpha\beta\neq 00}p_{\alpha\beta}$. The depolarizing channel has $p_{\alpha\neq 0}=p/3$ and $p_{\alpha\beta\neq 00}=p/15$, so that $p_0=1-p=p_{00}$. The probability of an error in this case, whether it is a single- or two-qubit channel, is then just $p$.

The situation of Pauli noise gives particularly simple noise behavior, with faulty swap gates effectively introducing Pauli errors into the embedded circuit in an independent manner. To see this, we consider the following situation of two computational qubits $c_1$ and $c_2$ sequentially coupled to a routing qubit $r$ by type-1 SWAPs:
\begin{center}
\begin{quantikz}[row sep=0cm, column sep =0.5cm]
\lstick{$r$}& \gate[swap]{}  & \qw&\qw\rstick{$c_1$} \\ 
\lstick{$c_1$}& &\gate[swap]{}&\qw\rstick{$c_2$}\\
\lstick{$c_2$}&\qw &&\qw\rstick{$r$}
\end{quantikz}
\end{center}
If both SWAPs are faulty, each subjected to the two-qubit noise channel $\cE_{q_1,q_1}$, straightforward calculation gives a product channel $\cE'_{c_1}\otimes \cE'_{c_2}$, where $\cE'_{q}$ is the single-qubit channel, 
\begin{equation}\label{eq:indepChannel}
\cE'_{q}(\,\cdot\,)\equiv\sum_{\alpha}{\Bigl(\sum_\beta p_{\alpha\beta}\Bigr)}\sigma_\alpha(\,\cdot\,)\sigma_\alpha
\end{equation}
as the noise on the output computational qubits (tracing away the routing qubit $r$). Faulty type-1 SWAPs thus only introduce Pauli errors into the embedded circuit in an independent manner. In particular, the errors on a computational qubit due to its participation in a faulty type-1 SWAP can equivalently be thought of as arising from a fault on that computational qubit location in the embedded circuit, with an added probability $p_{\textrm{1-SWAP},\alpha}\equiv\sum_\beta p_{\alpha\beta}$ of $\sigma_\alpha$ error, in accordance with Eq.~\eqref{eq:indepChannel}. For depolarizing noise on the swap gate, $p_{\textrm{1-SWAP},\alpha}=4p/15$, and the total additional probability of error is $\sum_{\alpha=1,2,3} p_{\textrm{1-SWAP},\alpha}=4p/5$.

We thus have a more concrete picture in this Pauli-noise situation, compared with the general scenario of Theorem \ref{theorem}, of how the errors due to swap gates can be attributed to faults in the computational qubits. Consider a swap sequence between two consecutive computational layers. Every time a computational qubit participates in a type-1 SWAP in that swap sequence, it acquires an additional $p_{\textrm{1-SWAP},\alpha}$ for the probability of a $\sigma_\alpha$ error. The errors from faulty type-2 SWAPs are lumped together with the associated (two-qubit) computational gate at the end of the swap sequence. Note that, since Pauli errors only either commute or anti-commute, the order in which the Pauli errors occur during the swap sequence do not matter; we need only keep track of the accumulated probability of each kind of Pauli error. The exact effective noise model thus depends on the details of the full routing schedule, and can be very complicated in general. 

\begin{table}
    \begin{tabular}{c|c}
         Pauli error& probability \\
         \hline
         no error $\id\otimes\id$ & $1-p-p_{\textrm{swap}}{\left[\frac{4}{5}(n_{1,1}+ n_{1,2})+n_2\right]}$ \\
         $\sigma_i\otimes \id$ & $\frac{1}{15}[p+p_{\textrm{swap}}(4 n_{1,1}+n_2)]$,\quad $i=1,2,3$ \\
         $\id\otimes \sigma_j$ & $\frac{1}{15}[p+p_{\textrm{swap}}(4n_{1,2}+n_2)]$, \quad $j=1,2,3$ \\
         $\sigma_i\otimes \sigma_j$ & $\frac{1}{15}(p+p_{\textrm{swap}} n_2)$,\quad $i,j=1,2,3$
    \end{tabular}
    \caption{\label{tab:effChannel} The effective noise channel on a CNOT in the embedded circuit. The probability expressions are accurate to linear order in $p$.}
\end{table}

For our surface-code examples under depolarizing noise, however, we can derive a simpler, albeit, approximate, description. We first observe that, for the SE circuits, all multi-qubit computational gates are two-qubit CNOTs. The CNOTs in each SE circuit occur in 4 computational layers, and every (abstract or computational) qubit, apart from the boundary ones, participates in exactly one CNOT in each layer. Between each pair of CNOT layers, we insert swap sequences to route together any pair of unconnected computational qubits that participate in the same CNOT. Type-2 SWAPs are hence allowed only between each such pair of computational qubits. For a particular long-distance CNOT, suppose the swap sequence needed to move a pair of computational qubits $q_1, q_2$ together has $n_{1,1}$ and $n_{1,2}$ type-1 SWAPs on $q_1$ and $q_2$, respectively, and $n_2$ type-2 SWAPs. Then, the effective noise of that CNOT in the embedded circuit can be described as given in Table \ref{tab:effChannel}. Here, we have taken into account the added noise from the swap sequence, as well as the original noise on the computational CNOT itself. The swap gates are assumed to be subjected to depolarizing noise with parameter $p_\textrm{swap}[\sim O(p)]$, anticipating that the SWAPs are possibly not native, but are made up of $O(1)$ primitive gates.  

From Table \ref{tab:effChannel}, we see that the effective noise channel for each CNOT (or, more generally, a two-qubit gate) in the embedded circuit is an asymmetric Pauli channel, with coefficients for the different Pauli errors that depend on the details of the routing schedule; the asymmetry arises from type-1 SWAPs. Nevertheless, one might hope for an approximate summary of the behavior by a single effective noise parameter $\peff'$ (where the prime on $\peff'$ distinguishes it from the best-fit $\peff$ used in the main text and Apps~\ref{app:HexEx}) using spacetime averages of the routing schedule parameters and ignoring the asymmetry. Specifically, we regard the probability of no error (the $\id\otimes\id$ term in Table \ref{tab:effChannel}) as the $p_0\equiv 1-\peff'$ parameter for a depolarizing channel, replacing $n_2$ and the sum of $n_{1,1}$ and $n_{1,2}$ by their averages (indicated by a bar) over all qubits and layers for the entire routing schedule:
\begin{equation}\label{eq:peffprime}
\peff'\equiv p+p_\textrm{swap}{\left[\tfrac{4}{5}\overline{n_1}+\overline{n_2}\right]},
\end{equation}
where $\overline{n_1}\equiv \overline{n_{1,1}}+\overline{n_{1,2}}$ counts the average number of type-1 SWAPs experienced by a computational qubit. While the single parameter $\peff'$ cannot capture the asymmetry present in the actual effective noise channel, for our examples of heavy-hexagonal and hexagonal lattices, it nevertheless provides a good summary---as confirmed by our simulations---of the noise deterioration introduced by our routing schedule, as we explain below.

In our simulations, each SWAP is composed from 3 CNOTs, with a two-qubit depolarizing noise (with strength $p$) on each CNOT. The noise channel associated with each SWAP can be easily shown to still be depolarizing noise, to linear order in $p$, but with the noise strength $p_\textrm{swap}=3p$. Furthermore, for the hexagonal-lattice example, where the optimal swap schedule employs only a single type-2 SWAP, that SWAP is always adjacent to a computational CNOT. Under the 3-CNOT decomposition of the type-2 SWAP, this gives 4 CNOTs in a row such that a pair of CNOTs always cancels out, leaving a remaining pair of CNOTs to implement both the type-2 SWAP and the computational CNOT (this cancellation argument applies generally for all routing schedules, when type-2 SWAPs are decomposed in terms of CNOTs). Allocating the noise from the computational gate to one of the two CNOTs, and the noise for the swap gate to the other, we can set $p_\textrm{swap}=p$ in this case. 

Now, for the hexagonal case, we have $\overline{n_1}=0$ and $\overline{n_2}=1/4$, with the division by $4$ accounting for the average over the 4 CNOT computational layers. For the heavy-hexagonal case, we have $\overline{n_1}=3.5/4$ and $\overline{n_2}=0$, with $3.5$ being the average number of type-1 SWAPs per SE round (with its 4 computational CNOTs) as stated in the main text (and see App.~\ref{app:HeavyEx}). Then, we have, for our two examples, the following $\peff'$ values:
\begin{align} \label{eq:p_eff_values}
    \textrm{heavy-hex.:}\quad \peff'&=p+(3p){\left[\tfrac{4}{5}\times 3.5/4+0\right]}=3.1p\nonumber\\
    \textrm{hex.:}\quad \peff'&=p+p{\left[0+1/4\right]}=1.25p.
\end{align}
These come close to the best-fit numerical $\peff$ values given in the main text, and Eq.~\eqref{eq:peffprime} can be used as a first estimate of how the noise strength gets rescaled by the EPP routing schedule, and the corresponding threshold deterioration.

\subsection{Syndrome decoding for embedded circuits} \label{sec:mwpm_correction}
Here, we describe how the error mechanisms of embedded circuits can be incorporated into the surface code syndrome decoding. As discussed in the main text, the swap gates in the physical circuit can introduce correlated errors in addition to existing errors in the abstract circuit. These correlated errors, while not altering the fault-tolerance properties of the code, can affect the accuracy of the decoder. This is particularly so since different decoders for surface codes compete largely in the space where there are more errors than  guaranteed correctable by the code distance, and hence beyond the fault-tolerance (and EPP) considerations. One expects to be able to reduce the logical error rates by making use of error-correlation information in the decoder. We explain below how to incorporate that into the matching graph of a minimum-weight perfect-matching (MWPM) surface code decoder; the ideas can be extended to other codes decoded in a similar manner.

We consider a distance-$d$ surface code that corrects up to $t = \left\lfloor d/2 \right\rfloor$ errors. We assume it has been embedded into the hardware using an EPP routing schedule. The decoder sees only the embedded circuit, not the physical circuit; in particular, the matching graph---which carries the information about where single errors can occur in the circuit and with what weight (probability), and is used by the MWPM algorithm to decide on the allowed matchings and corresponding weights---contains only nodes and edges that represent data and ancilla qubits of the embedded circuit. The matching graph is constructed as follows:
\begin{enumerate}
    \item Begin with the standard matching graph \cite{dennis2002topological,fowler2012surface} for the abstract circuit, to capture the basic error mechanisms present even without routing.
    \item Types-1 and 2 SWAPs that fail independently mimic existing error mechanisms of the abstract circuit, and thus already exist as edges in the matching graph. Incorporate these additional sources of errors by locally updating their corresponding edge weights, which yields noise strength that are additive to the leading order.
    \item Type-1 SWAPs that connect multiple computational qubits to the same routing qubit can generally lead to correlated errors. These can be incorporated by decomposing them in terms of their constituent single-qubit error mechanisms. Explicitly, all correlated error mechanisms of weight $s \leq t$ are identified, which arise whenever the same routing qubit is involved in $s$ type-1 SWAPs. For each of them, trace out the routing qubits to yield a spatio-temporally correlated error mechanism involving $s$ locations. The noise strength of this error mechanism is simultaneously added to its constituent single-qubit error mechanisms, each of which already exists as an edge in the matching graph. 
\end{enumerate}
This procedure yields a matching graph with the same connectivity as that of the surface code, with updated edge weights, which can be used for syndrome decoding via the MWPM algorithm.

We note that the correlated errors can be directly included as hyper-edges (i.e., edges joining more than 2 nodes) in the matching graph, instead of decomposing them in terms of independent error mechanisms. This improves the performance of decoding and thus lowers logical error probabilities and threshold requirements, at the expense of transforming the matching graph into a hypergraph that requires a hypergraph matching algorithm \cite{higgott2023improved,chen2022calibrated}. Specifically, correlated errors involving $s$ qubits or fault locations appear as order-$2s$ hyper-edges (i.e., edges joining $2s$ nodes) involving existing nodes in the matching graph that can always be formed by the union of nodes involved in $s$ existing edges (each corresponding to an existing single-qubit error mechanism that a faulty type-1 SWAP mimics). In our simulations, we adopt the simpler approach of working only with MWPM and a matching graph consisting of graph-like edges (i.e., edges only join 2 nodes).

\subsection{Further details of the heavy-hexagonal example} \label{app:HeavyEx}

\begin{figure}
    \centerline{\includegraphics[width=.5\textwidth]{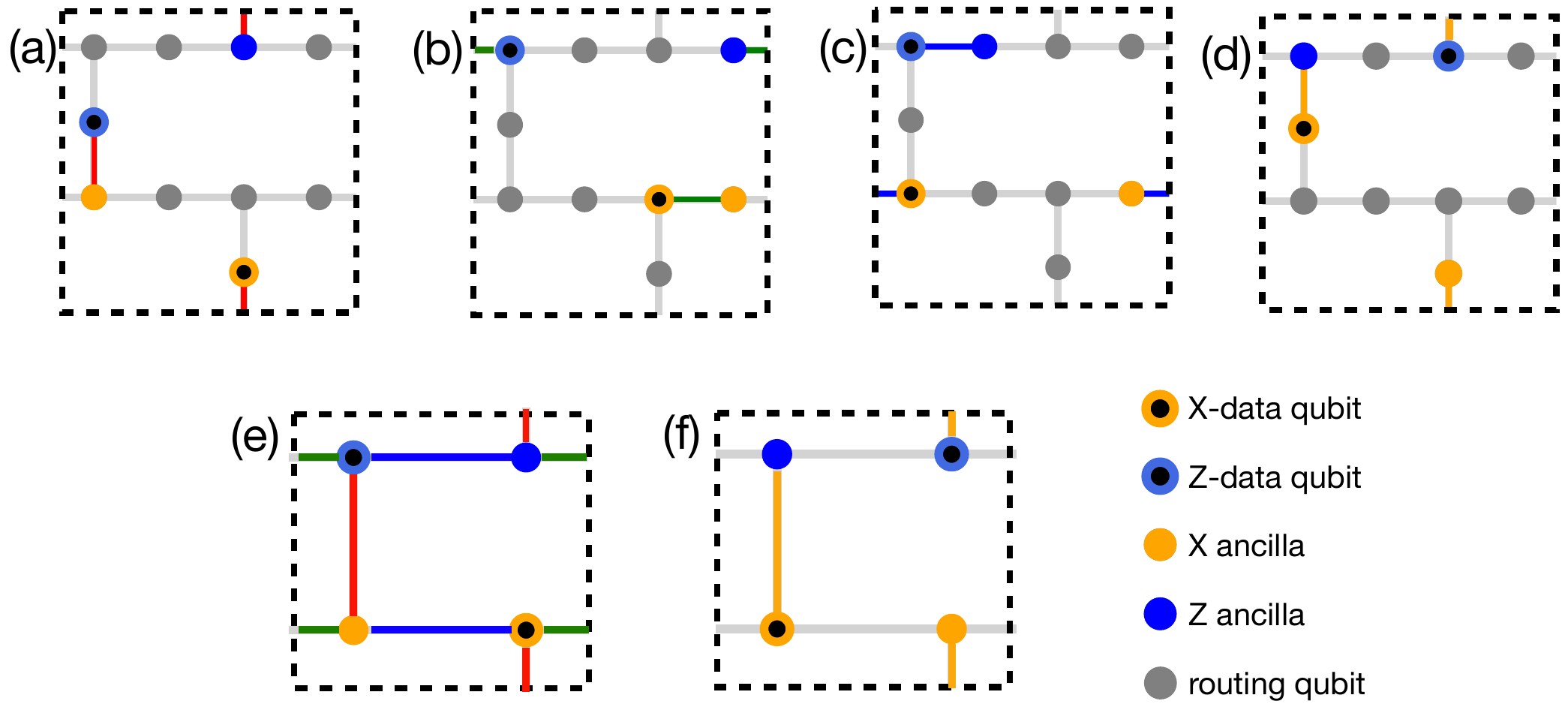}}
    \caption{Intermediate embeddings for the heavy-hexagonal [(a)--(d)] and hexagonal [(e) and (f)] lattice EPP schedules. Red, green, blue, and yellow colored edges denote computational operations corresponding to the CNOT layers $C_1, C_2, C_3$, and $C_4$, respectively.
    }
    \label{fig:schedule_schematic}
\end{figure}

Here, we describe the routing schedule used in our heavy-hexagonal lattice example, as found by the greedy distance-minimizing routing algorithm described in App.~\ref{app:Routing}. The precise moves of the EPP routing schedule are given in Table.~\ref{table:hhex_schedule}, with intermediate embeddings during CNOT layers illustrated in Fig.~\ref{fig:schedule_schematic} (a)-(d).
It consists of 5 SWAP layers interspersed between the 4 computational CNOT layers of one round of syndrome extraction. At the end of one full round, the computational qubits are displaced one unit cell diagonally from their initial positions, which are reversed in the next round by reversing the ordering and directions of the swap moves.

\begin{table}[!h]
\begin{center}
\begin{tabular}{|l|c|c|c|c|c|c|c|c|c|}
\hline
             & $C_1$ & $S_1$ & $C_2$ & $S_2$ & $S_3$ & $C_3$ & $S_4$ & $S_5$ & $C_4$ \\ \hline
$X$ data     &$d$ &$u$&$r$&$l$&$l$ &$l$&$u$&$i$ &$u$ \\ \hline
$Z$ data     &$d$ &$u$&$l$&$i$&$i$ &$r$&$l$&$l$ &$u$ \\ \hline
$X$ ancilla  &$u$ &$l$&$l$&$i$&$i$ &$l$&$l$&$d$ &$d$ \\ \hline
$Z$ ancilla  &$u$ &$r$&$r$&$l$&$l$ &$r$&$i$&$l$ &$d$ \\ \hline
\end{tabular}
\caption{\label{table:hhex_schedule} 5-SWAP-layer routing schedule for the heavy-hexagonal lattice example. The $C_i$ columns are the computational CNOT layers while the $S_i$ columns are the swap layers. For the $S_i$ columns, the symbols $u$, $d$, $l$, $r$, and $i$ indicate swaps of the qubit (row) with its neighboring qubit to move up, down, left, right, or idle respectively. For the $C_i$ columns, the same symbols denote the direction of the qubit interacting with it through a CNOT gate (i.e. the orientations of the colored edges in Fig.~\ref{fig:schedule_schematic}).}
\end{center}
\end{table}

Each data or ancillary qubit experiences, on average, $\overline{n_1} = 3.5$ and $\overline{n_2} = 0$ type-1 and type-2 SWAPs respectively. Decomposing each swap gate as 3 CNOT gates, we find that it can be executed within 19 timesteps. Notably, certain routing qubits in this schedule are involved in type-1 SWAPs with computational qubits that are not connected in the interaction graph, giving rise to the possibility of spatio-temporally correlated errors as mentioned in the main text.

\subsection{Further details of the hexagonal example}
\label{app:HexEx}

\begin{table}
\begin{tabular}{|l|c|c|c|c|c|}
\hline
  & $C_1$ & $C_2$ & $C_3$ & $S_1$ & $C_4$ \\ \hline
$X$ data     & $d$ & $r$ & $l$ & $r$ & $u$    \\ \hline
$Z$ data     & $d$ & $l$ & $r$ & $r$ & $u$   \\ \hline
$X$ ancilla  & $u$ & $l$ & $r$ & $r$ & $d$    \\ \hline
$Z$ ancilla  & $u$ & $r$ & $l$ & $r$ & $d$    \\ \hline
\end{tabular}
\caption{\label{table:hex_schedule} 1-SWAP-layer routing schedule for the hexagonal lattice. Notations follow that of Table.~\ref{table:hhex_schedule}.}
\end{table}

\begin{figure}
    \centerline{\includegraphics[width=.5\textwidth]{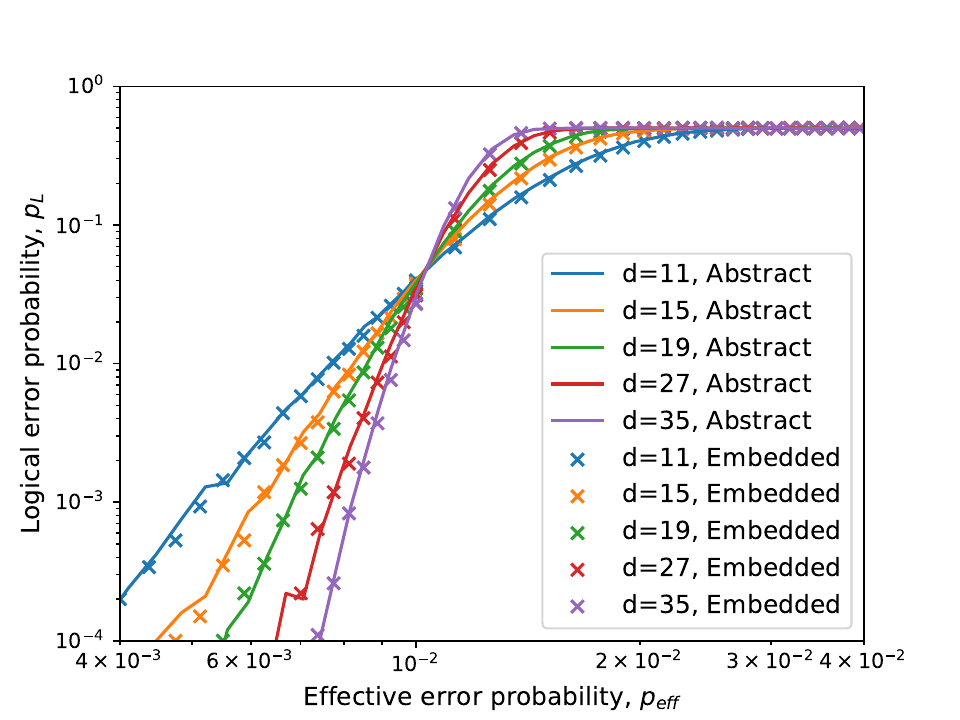}}
    \caption{The logical error probability $\pL$ versus the effective error probability $\peff$ for the example of the surface code embedded onto the hexagonal lattice under depolarizing noise with error probability $p$. For the benchmark noisy abstract circuit, $\peff=p$; for the physical (embedded) circuit, $\peff=1.25p$ obtained by fitting both curves for $d=35$ below the threshold value.
    }
    \label{fig:hexacancel_LER}
\end{figure}

Here, we describe the routing schedule used in our hexagonal lattice example, as found by the greedy distance-minimizing routing algorithm described in App.~\ref{app:Routing}.
The precise moves of the EPP routing schedule are given in Table.~\ref{table:hex_schedule}, with intermediate embeddings during CNOT layers illustrated in Fig.~\ref{fig:schedule_schematic} (e) and (f). It consists of a single SWAP layer between the final pair of interaction layers, consisting exclusively of type-2 swap gates. There are also no routing qubits involved -- every physical qubit is mapped to a data or ancilla qubit of the surface code. At the end of one full round, the computational qubits are displaced one unit cell horizontally from their initial positions, which are reversed in the next round by reversing the ordering and directions of the swap moves. 

Numerical simulations of the logical error probabilities are given in Fig.~\ref{fig:hexacancel_LER}. Each data or ancillary qubit experiences $\overline{n_1} = 1$ swap gates on average. Decomposing each swap gate as 3 CNOT gates and simplifying the circuit via cancellations of two CNOT gates as described in App.~\ref{sec:approx_noise_model_surf}, we find that it can be executed within 5 timesteps. Notably, this schedule saturates the naive lower bound of $1$ SWAP layer, necessary when embedding a degree-$4$ interaction graph onto a degree-$3$ device graph, and is therefore optimal in the number of swap layers. In addition,  since no routing qubits are involved, it is also optimal in terms of the number of physical qubits required to embed the surface code onto the hexagonal lattice.

\end{document}